\documentclass[conference, anonymous]{IEEEtran}
\usepackage{graphicx} 
\usepackage{authblk}
\usepackage[margin=0.8in]{geometry}
\usepackage{cite}
\usepackage{amsmath,amssymb,amsfonts,amsthm}
\usepackage{algorithmic}
\usepackage{graphicx}
\usepackage{textcomp}
\usepackage{xcolor}
\usepackage{hyperref}
\def\BibTeX{{\rm B\kern-.05em{\sc i\kern-.025em b}\kern-.08em
    T\kern-.1667em\lower.7ex\hbox{E}\kern-.125emX}}

\title{Universally Composable Commitments with Communicating Malicious Physically Uncloneable Functions}

\author{Louren\c co Abecasis$^{1,2}$, Paulo Mateus$^{1,2}$ and Chrysoula Vlachou$^{1,2}$}
\affil{$^1$Instituto de Telecomunica\c c\~oes\\ Av. Rovisco Pais 1, 1049-001 Lisboa, Portugal}
	\affil{$^2$Department of Mathematics, Instituto Superior T\' ecnico\\
		Universidade de Lisboa, Av. Rovisco Pais 1, 1049-001 Lisboa, Portugal}

\usepackage[
    n,
    advantage,
    operators,
    sets,
    adversary,
    landau,
    probability,
    notions,
    logic,
    ff,
    mm,
    primitives,
    events,
    complexity,
    oracles,
    asymptotics,
    keys] {cryptocode}

\def\bm#1{\mathchoice                             
  {\mbox{\boldmath$\displaystyle#1$}}%
  {\mbox{\boldmath$#1$}}%
  {\mbox{\boldmath$\scriptstyle#1$}}%
  {\mbox{\boldmath$\scriptscriptstyle#1$}}}

\newcommand{\prn}[1]{\ensuremath{\left( #1 \right)}}
\newcommand{\msf}[1]{\ensuremath{\mathsf{#1}}}
\newcommand{\ensemble}[1]{
    \ensuremath{ \left\{ #1 \right\}_{n \in \mathbb{N}} }
}

\createpseudocodeblock{pcb}{ center, boxed }{}{}{}
\createprocedureblock{procb}{center, boxed}{}{}{}

\def \N {\ensuremath{\mathbb{N}}}

\newcommand{\indep}{\perp \!\!\! \perp}
\newcommand{\maxentropy}[1]{\ensuremath{\operatorname{H_0} \prn{ #1 } }}

\def \FMPUF {\ensuremath{\mathcal{F}_{\mathsf{MPUF}}}}
\def \FComMPUF {\ensuremath{\mathcal{F}_{\mathsf{ComMPUF}}}}

\def \Fcom {\ensuremath{\mathcal{F}_{\mathsf{com}}}}

\def \Rideal {\ensuremath{\widetilde{\mathsf{R}}}}
\def \Sideal {\ensuremath{\widetilde{\mathsf{S}}}}

\def \A {\ensuremath{A}}

\def \Z {\ensuremath{\mathcal{Z}}}
\def \Adv {\ensuremath{\mathcal{A}}}
\def \Sim {\ensuremath{\mathcal{S}}}

\def \puffamily {\ensuremath{\mathcal{P}}}
\def \pufsample {\msf{Sample}}
\def \pufeval {\msf{Eval}}
\def \rg {\ensuremath{rg}}
\def \dham {\msf{d}}
\def \dnoise {\msf{d_{noise}}}
\def \dmin {\msf{d_{min}}}
\def \entropybound {\ensuremath{m}}

\def \PUF {\ensuremath{\mathsf{PUF}}}
\def \MPUF {\ensuremath{\mathsf{MPUF}}}
\def \ComMPUF {\ensuremath{\mathsf{ComMPUF}}}
\def \TQ {\ensuremath{\mathsf{TQ}}}

\def \fe {\texttt{FE}}
\def \Gen {\ensuremath{\mathsf{Gen}}}
\def \Rep {\ensuremath{\mathsf{Rep}}}
\def \fem {\ensuremath{m}}
\def \fel {\ensuremath{\ell}}
\def \fet {\ensuremath{t}}
\def \feeps {\ensuremath{\epsilon}}

\def \extra {\msf{extra}}
\def \EXTRA {\msf{EXTRA}}
\def \SD {\msf{SD}}

\def \CPUF {\texttt{CPUF}}
\def \CUnif {\texttt{CUnif}}
\def \CBS {\texttt{CBS}}

\def \ExtPUF {\texttt{ExtPUF}}

\def \CollExtPUF {\texttt{CollExtPUF}}
\def \CollExtUnif {\texttt{CollExtUnif}}
\def \CollExtBS {\texttt{CollExtBS}}

\def \CS {\ensuremath{\mathsf{1}}}
\def \CR {\ensuremath{\mathsf{2}}}
\def \E {\ensuremath{\mathsf{E}}}
\def \PUFCS {\ensuremath{\PUF_\CS}}
\def \PUFCR {\ensuremath{\PUF_\CR}}
\def \PUFE {\ensuremath{\PUF_\E}}

\def \STATE {\ensuremath{\msf{STATE}}}

\def \Enc {\msf{Enc}}
\def \Dec {\msf{Dec}}

\def \bs {\ensuremath{bs}}

\def \kin {\ensuremath{k_\msf{in}}}
\def \kout {\ensuremath{k_\msf{out}}}
\def \kstate {\ensuremath{k_\msf{state}}}

\def \R {\ensuremath{\mathsf{R}}}
\def \S {\ensuremath{\mathsf{S}}}
\def \Radv {\ensuremath{\mathsf{R^*}}}
\def \Sadv {\ensuremath{\mathsf{S^*}}}
\def \ext {\ensuremath{\mathsf{E}}}

\def \C {\ensuremath{\mathcal{C}}}
\def \D {\ensuremath{\mathcal{D}}}

\def \com {\textup{\texttt{Com}}}
\def \collcom {\textup{\texttt{CollCom}}}
\def \commit {\textup{\texttt{Commit}}}
\def \INTER {\textup{\texttt{INTER}}}
\def \STRINGS {\textup{\texttt{STRINGS}}}
\def \Const {\msf{Const}}
\def \OPEN {\textup{\texttt{OPEN}}}
\def \CLOSED {\textup{\texttt{CLOSED}}}
\def \IND {\textup{\texttt{IND}}}

\newtheorem{theorem}{Theorem}
\newtheorem*{theorem*}{Theorem}
\newtheorem{definition}{Definition}
\newtheorem{lemma}{Lemma}
\newtheorem{proposition}{Proposition}
\newtheorem{remark}{Remark}
\date{January 2025}

\begin{document}

\maketitle

\begin{abstract}
In this work, we explore the possibility of universally composable (UC)-secure commitments using Physically Uncloneable Functions (PUFs)
within a new adversarial model.
We introduce the \textit{communicating malicious PUFs}, i.e. malicious PUFs that can interact with their creator even when not in
their possession, obtaining a stronger adversarial model.
Prior work [ASIACRYPT 2013, LNCS, vol. 8270, pp. 100--119] proposed a compiler for constructing UC-secure commitments from ideal extractable commitments, and our task would be to adapt the ideal extractable commitment scheme proposed therein to our new model.
However, we found an attack and identified a few other issues in that construction, and
to address them, we modified the aforementioned ideal extractable commitment scheme and introduced new properties and tools that allow us to rigorously develop and present security proofs in this context. 
We propose a new UC-secure commitment scheme against adversaries that can only create
stateless malicious PUFs which can receive, but not send, information from their creators. Our protocol is more efficient compared to previous proposals, as we have parallelized the ideal extractable commitments within it.
The restriction to stateless malicious PUFs is significant, mainly since the protocol from [ASIACRYPT 2013, LNCS, vol. 8270, pp. 100--119] assumes malicious PUFs with unbounded state, thus limiting its applicability. However it is the only way we found to address the issues of the original construction. We hope that in future work this restriction can be lifted, and along the lines of our work, UC-secure commitments with
fewer restrictions on both the state and communication can be constructed.

\end{abstract}


\section{Introduction}
The universally composable (UC) security framework, introduced in \cite{CanettiEarly},
is a powerful standard ensuring that cryptographic protocols remain secure when they are composed with
other protocols and/or instances of themselves. 
UC-security has been extensively explored in the context of secure multi-party computation (MPC) \cite{UCCanetti,UCCommitmentsCanetti,UC2PMP,UCLimitations,UCGlobSetup,UCMPC,CanettiLate}, and it has been shown that in  the plain model, where no additional
assumptions are made, UC-secure MPC is impossible \cite{UCCommitmentsCanetti}. Therefore, to achieve this compelling composability property one needs to make further assumptions, see e.g. access to a common reference string \cite{UCCommitmentsCanetti}. In our work, we employ hardware assumptions, i.e. assumptions based on the physical properties of specific hardware components. In particular, we consider Physically Uncloneable Functions (PUFs), introduced in \cite{Pappu} and used in a line of successive works for constructing UC-secure protocols for MPC primitives \cite{PUFsUC,MalPUFs,UCComm,FeasibilityPUFs,BoundedStateOT,EverlastingUCCom}. 
A PUF is a device produced by a complex physical manufacturing
process, making it extremely difficult to clone. It can be evaluated by providing a physical stimulus, called the \emph{challenge}, to which it responds with a noisy output, called the \emph{response}. Due to their uncloneability, PUFs were initially used as hardware tokens for device identification and authentication, however additional security properties  have been considered and to date there is a vast bibliography on how to incorporate them in diverse security applications (see e.g. \cite{PUFsurvey} for a review).

\subsection{Prior related work}
Focusing on the works concerning the use of PUFs for UC-secure MPC,
the first is the one by Brzuska et al. \cite{PUFsUC}, where UC-secure
protocols for oblivious transfer (OT), bit commitment and key exchange were proposed. 
The PUFs are assumed to be \emph{trusted}, i.e. they have been produced through the prescribed manufacturing process and they have not been tampered with by an adversary. This assumption was later lifted in \cite{MalPUFs}, where \emph{malicious} PUFs were introduced to account for adversaries that can create PUFs with arbitrary
malicious behaviour. A malicious PUF, in general, is a hardware token that meets the syntactical requirements of an honest PUF. It could be a fake PUF, possibly programmed with malicious
code, or a PUF whose output on some input might depend on previous inputs. The latter is called a \emph{stateful} PUF and is in contrast with \emph{stateless} PUFs. While honest PUFs are necessarily stateless, malicious PUFs can be either stateful or stateless. In this model of stateful malicious PUFs, and assuming that a malicious PUF cannot interact with its creator
once it is sent away to another party, Ostrovsky et al., developed a computationally UC-secure commitment scheme \cite{MalPUFs}.  
They also proposed a commitment scheme with unconditional indistinguishability-based security
 in the malicious PUF model, as well as an unconditionally UC-secure OT protocol in the so-called \emph{oblivious-query model}, where the adversaries cannot create malicious PUFs, but they can query trusted
PUFs via non-prescribed processes. They also show that UC-security is
impossible when considering adversaries from both the malicious PUFs and the oblivious-query models. Subsequently, Damgård and Scafuro  constructed an
unconditionally UC-secure commitment scheme using the same model  from \cite{MalPUFs} for malicious PUFs \cite{UCComm}. They also showed that these commitments are unconditionally UC-secure in the stateless tamper-proof hardware token model \cite{KatzTokens}. Following this, Dachman-Soled et al.  derived also two important results \cite{FeasibilityPUFs}: the first is the impossibility of unconditionally secure OT, for both stand-alone and indistinguishability-based security, in the stateful malicious PUF model, and the second is the possibility of UC-secure OT in the stateless malicious PUF model. 
In all the works mentioned so far, malicious PUFs are assumed to maintain a priori unbounded states, however Badrinarayanan et al. 
 argued that this is a very strong assumption that might be practically irrelevant \cite{BoundedStateOT}. Therefore, they modified the adversarial model such that the malicious stateful PUF can maintain an
a priori bounded state, and showed that unconditional UC-secure computation of any functionality is possible in this model, employing the construction from \cite{UCComm}. Moreover, they introduced a new model, where the adversary can generate malicious stateless PUFs and encapsulate honest PUFs inside them even without the knowledge of the
functionality of the inner PUFs. The outer malicious PUF can
make oracle calls to the inner PUFs, and 
an honest party  is not able to tell whether they are interacting with an honest
PUF or a malicious PUF encapsulating honest ones. In this \emph{malicious encapsulation model}, they show that unconditional UC-secure computation of any functionality is still possible.
Finally, Magri et al.  introduced a more general and stronger model, that of \emph{fully malicious hardware tokens}, of which PUFs are a special case  \cite{EverlastingUCCom}. Such tokens do not have a priori
bounded states, arbitrary code might be installed inside them, and
they can encapsulate and decapsulate other – possibly fully
malicious – tokens within themselves. An everlastingly UC-secure commitment scheme was constructed under the Learning With Errors assumption, and it was shown that
everlastingly UC-secure OT is impossible using non-erasable honest tokens.

\subsection{Our contributions}
\label{subsec:contibutions}
We propose a new  model for malicious PUFs, the \emph{communicating malicious PUFs}, and explore the possibility of unconditionally UC-secure commitments in this model. The existence of such commitments is essential for various MPC tasks, and the methods in \cite{UCComm} that are used to derive them are very relevant in this context. Along these lines, in the quantum cryptography paradigm, i.e. assuming access to quantum channels, unconditionally UC-secure commitments imply quantum UC-secure OT \cite{QuantumMPC},  making the scheme in \cite{UCComm} very important in this respect as well, as it enables quantum UC-secure MPC relying solely on physical assumptions. Our motivation was to probe the limitations of this construction by strengthening the adversarial model. In particular, we allow a malicious PUF to communicate with its creator even when it is not in their hands and we aim to determine whether and under which conditions unconditional UC-security still holds. We believe that considering malicious communicating PUFs is a reasonable assumption not only from a theoretical point of view, as an extension of the adversarial model, but also from a practical point of view, since hardware realizations do not rule out this possibility.  
In the previous works, it was assumed that malicious PUFs do not communicate with their creator when they are sent away
to another party, as it is argued that if the functionality allowed this, then the model would be equivalent
to the plain UC model \cite{MalPUFs}, where UC-secure computation is impossible \cite{UCCommitmentsCanetti}.
However, this only holds if the communication is unbounded, and we believe it is relevant to consider \emph{bounded} communication and study the possibility of UC-secure commitments in this case. 
Following the original approach in \cite{UCComm}, we construct an extractable ideal (i.e., statistically hiding and binding) commitment scheme and,  using an adapted version of the unconditional black-box compiler developed therein, we obtain a UC-secure protocol for commitments in the communicating malicious PUFs model.  Note that our protocol is stated for the UC-secure commitment of bitstrings and not bits, which was the case in previous works.
We present our new model and the corresponding communicating malicious PUFs functionality in Section \ref{sec:ourmodel}, and  our proposal for UC-secure commitments in Section \ref{sec:ourprotocol}. Before this, in Sections 
\ref{sec:fixextractability} and \ref{sec:fixcompiler}, we address a few issues that we encountered while going through  previous works, and are essential for making these constructions rigorous. In particular, first we noticed that a malicious sender could break the extractability property of the ideal extractable commitment protocol from \cite{UCComm}. In Section \ref{subsec:issues}, we present this attack and discuss how to fix it. Succinctly, one could either change the extractor or the protocol, and since we could not find a way to change the extractor without giving excessive power, we changed the protocol. However, this change comes at the cost of having to assume that the malicious PUFs  are stateless, a restriction which is very strong in our view and which we would like to lift in future work; especially because the original protocol was designed to be secure against adversaries that can create malicious PUFs with unbounded states. Importantly, the same approach as in \cite{UCComm} is followed in \cite{BoundedStateOT}, where considering stateful malicious PUFs is essential.  Therefore, the attack that we found and, in turn, the way to fix it influence the results in \cite{BoundedStateOT}, as well. We also revised some of the PUF properties from previous works and introduce  new ones in order to rigorously develop and state our results and proofs. The revised and new PUF properties along with the reasons for modifying them are presented in Section \ref{subsec:additionalproperties}. Moreover, we noticed another issue in \cite{UCComm}: the UC-secure commitment protocol involves multiple commitments, and 
 in the UC-security proof it is implicitly assumed that the security properties of these commitments are preserved when they are employed collectively within a more complex protocol. This assumption, though, was not proven. To avoid possible flaws in the proof and to create a more efficient UC-secure protocol, through the parallelization of these commitments, we generalized 
the definition of a commitment scheme and its corresponding properties to enable the commitment of many strings at once, and adjusted accordingly the compiler from \cite{UCComm}. We should mention that our notation overall is different from that in previous works, as we wanted it to be generalizable for the purpose of this collective commitments scheme.

Even though the protocol we propose in our new adversarial model faces severe restrictions, i.e. the malicious PUFs have to be stateless and have no outgoing communication, that was the only way we found  to make it work against the attack we discovered for the ideal extractable commitments in \cite{UCComm}. Hopefully, our contribution will trigger further work on overcoming this restriction, and using the new tools and approach we propose here, UC-secure commitment schemes, possibly even more efficient, will be developed in strong adversarial models with less restrictions on the state and communication. 

\section{Our adversarial model: communicating malicious PUFs}
\label{sec:ourmodel}

Before presenting the adversarial model, we should mention that, just like in previous works, all adversaries are Probabilistic Polynomial-Time (PPT).

Let us start by briefly describing the malicious PUFs model without communication that we build upon, as it slightly differs from those in previous works.
Malicious PUFs were first introduced in \cite{MalPUFs} to model adversaries that can tamper with the manufacturing process of PUFs,
potentially embedding additional behaviors, such as query logging, into the PUF tokens.
To keep the model as general as possible, \cite{MalPUFs} places no restrictions on the malicious PUF families other than requiring
that they share the same syntax as honest PUF families.
In addition, the malicious PUF functionality is parameterized by both an honest and a malicious PUF family.
We argue, however, that this approach grants the adversary excessive power due to its lack of specificity.
For instance, without restrictions on the malicious PUF family, an adversary could, in an extreme case, create a PUF that replicates
the most recently generated honest PUF, which would violate uncloneability -- a fundamental property of PUFs.
Furthermore, as pointed out in \cite{FeasibilityPUFs}, restricting malicious PUFs to a fixed family prevents them from being
created adaptively throughout the protocols. To address these concerns, \cite{FeasibilityPUFs} proposes a more explicit model in which malicious PUFs are defined by arbitrary code,
potentially including oracle access to a freshly created honest PUF that remains inaccessible to other parties, and we follow this modelling as well.
 However, in \cite{FeasibilityPUFs} each honest PUF was assumed to generate a random response on the first query for each challenge and return the same response for repeated queries. We believe that this simplification might not be realistic\footnote{Consistency is not guaranteed when using fuzzy extractors (see Definition \ref{def:fuzzyextractor}): if we query $\PUF(s)$ and receive different responses, $\sigma$ and $\sigma'$, the resulting outputs $st$ and $st'$ from the fuzzy extractor may not match.}, therefore we chose to retain the original approach for modeling honest PUFs.

We model malicious PUFs, denoted as \MPUF{},  as follows: each \MPUF{} consists of a finite set\footnote{Since we consider PPT adversaries, the size of this set must be polynomially bounded in the security parameter $n$.} of freshly created honest PUFs, that are inaccessible to other parties and a possibly stateful Turing machine $M$ with oracle access to those PUFs. Querying \MPUF{} on a challenge $s$ thus amounts to querying $M$ on $s$, which may change $M$’s state. For security parameter $n$, the machine $M$ operates with $\kstate (n)$ bits of memory. The state size can be:
\begin{itemize}
    \item \textbf{bounded,} in which case it is represented as a function $\kstate:\mathbb{N}\rightarrow\mathbb{N}$, or
    \item \textbf{unbounded, }in which case $\kstate=\infty$.
\end{itemize}
Notice that, for simplicity, our malicious PUFs have access to
a single honest PUF, however the functionality can be easily generalized to account for access to 
multiple honest PUFs.
The corresponding functionality \FMPUF{} is depicted in Fig. \ref{fig:mpuf functionality} in Appendix
\ref{appendix:maliciousfunctionality}.

We can now proceed to our new adversarial model 
 which allows each communicating malicious PUF and its creator $P$ to communicate in a possibly bounded manner, and we denote it as \ComMPUF. 
We define two types of communication:
\begin{itemize}
    \item \textbf{Incoming communication:}
    Information sent from $P$ to \ComMPUF{}, using at most a total of $\kin(n)$ bits.
    This can potentially change \ComMPUF{}'s state and cause it to reply with a message.

    \item \textbf{Outgoing communication:}
    Information sent from \ComMPUF{} to $P$, using at most a total of $\kout(n)$ bits.
    This can happen when \ComMPUF{} is queried or, as mentioned earlier, when it receives a message from $P$.
\end{itemize}

Note that  $\kin(n)$ and  $\kout(n)$ depend on the security parameter $n$, and refer to the total communication  and not the size of each message, $m_\msf{in}$ and $m_\msf{out}$.
The corresponding functionality is depicted in Fig. \ref{fig:puf functionality}, and it is parameterized by a PUF family
\puffamily{} (see Definition \ref{def:PUFamily})  and  $\kin(n)$, $\kout(n)$ and $\kstate(n)$, which is the bound on the size of the state of the potentially  stateful communicating malicious PUF.
Each type of communication is tracked with a counter
and once the corresponding bound is reached, no further communication of that type occurs. For simplicity, we do not include the counter in the functionality. Finally, just like in the case without communication, we only consider access of the malicious PUF to a single honest PUF, but we can easily generalize the functionality.

\begin{figure}[ht]
    \centering
    \noindent\fbox{%
        \parbox{0.95\linewidth}{
            \begin{center}
                \textbf{Communicating Malicious PUF Functionality $\mathcal{F}_\ComMPUF(\puffamily, \kstate, \kin, \kout)$}
            \end{center}
        
            Run with parties $\mathbb{P} = \set{P_1, \cdots, P_k}$ and adversary \Sim{}.
            Create empty lists $\mathcal{L}$ and $\mathcal{M}$.
            \begin{itemize}
                \item Upon receiving $(\msf{sid}, \msf{init}, \msf{honest}, P)$ or $(\msf{sid}, \msf{init}, \msf{malicious}, M, P)$ from
                $P \in \mathbb{P} \cup \set{\Sim}$, check whether $\mathcal{L}$ contains some $(\msf{sid}, *, *, *, *)$:
                \begin{itemize}
                    \item If so, turn to the waiting state;
                    \item Else, draw $\msf{id} \gets \pufsample_n$, add $(\msf{sid}, \msf{honest}, \msf{id}, P, \bot)$
                    to $\mathcal{L}$ and send $(\msf{sid}, \msf{initialized})$ to $P$.
                    Furthermore, in the second case, add $(\msf{sid}, P, M)$ to $\mathcal{M}$.
                \end{itemize}

                \item Upon receiving $(\msf{sid}, \msf{eval}, P, s)$ from $P \in \mathbb{P} \cup \set{\Sim}$, check whether $\mathcal{L}$ contains
                $(\msf{sid}, \msf{mode}, \msf{id}, P, \bot)$ or $(\msf{sid}, \msf{mode}, \msf{id}, \bot, *)$ in case $P = \Sim{}$:
                \begin{itemize}
                    \item If it is not the case, turn to the waiting state;
                    \item Else, if $\msf{mode} = \msf{honest}$, run $\sigma \gets \pufeval_n \prn{\msf{id}, s}$
                    and send $(\msf{sid}, \msf{response}, s, \sigma)$ to $P$;
                    \item Else, if $\msf{mode} = \msf{malicious}$, get $(\msf{sid}, \tilde{P}, M)$ from $\mathcal{M}$ and run
                    $(\sigma, m_\msf{out}) \gets M (\msf{input}, s)$.
                    Then, send $(\msf{sid}, \msf{response}, s, \sigma)$ to $P$ and $(\msf{sid}, \msf{outmsg}, m_\msf{out})$ to $\tilde{P}$.
                \end{itemize}

                \item Upon receiving $(\msf{sid}, \msf{inmsg}, P, m_\msf{in})$ from $P \in \mathbb{P} \cup \set{\Sim}$, check whether $\mathcal{M}$ contains
                some $(\msf{sid}, P, M)$:
                \begin{itemize}
                    \item If it is not the case, turn to the waiting state;
                    \item Else, run $m_\msf{out} \gets M (\msf{msg}, m_\msf{in})$ and send $(\msf{sid}, \msf{outmsg}, m_\msf{out})$ to $P$.
                \end{itemize}

                \item Upon receiving $(\msf{sid}, \msf{handover}, P_i, P_j)$ from $P_i$, check whether $\mathcal{L}$
                contains some $(\msf{sid}, *, *, P_i, \bot)$:
                \begin{itemize}
                    \item If it is not the case, turn to the waiting state;
                    \item Else, replace the tuple $(\msf{sid}, \msf{mode}, \msf{id}, P_i, \bot)$ in $\mathcal{L}$ with
                    $(\msf{sid}, \msf{mode}, \msf{id}, \bot, P_j)$ and send $(\msf{sid}, \msf{invoke}, P_i, P_j)$ to \Sim{}.
                \end{itemize}

                \item Upon receiving $(\msf{sid}, \msf{ready}, \Sim{})$ from \Sim{}, check whether $\mathcal{L}$ contains
                $(\msf{sid}, \msf{mode}, \msf{id}, \bot, P_j)$:
                \begin{itemize}
                    \item If it is not the case, turn to the waiting state;
                    \item Else, replace the tuple $(\msf{sid}, \msf{mode}, \msf{id}, \bot, P_j)$ in $\mathcal{L}$ with
                    $(\msf{sid}, \msf{mode}, \msf{id}, P_j, \bot)$, send $(\msf{sid}, \msf{handover}, P_i)$ to $P_j$ and add
                    $(\msf{sid}, \msf{received}, P_i)$ to $\mathcal{L}$.
                \end{itemize}

                \item Upon receiving $(\msf{sid}, \msf{received}, P_i)$ from \Sim{}, check whether $\mathcal{L}$ contains that tuple.
                \begin{itemize}
                    \item If so, send $(\msf{sid}, \msf{received})$ to $P_i$;
                    \item Otherwise, turn to the waiting state.
                \end{itemize}
            \end{itemize}
        }
    }%
    \caption{The communicating malicious PUF functionality, $\mathcal{F}_\ComMPUF$.}
    \label{fig:puf functionality}
\end{figure}


\section{the ideal extractable commitment from \cite{UCComm}}
\label{sec:fixextractability}
\subsection{PUF properties: new and revised}
\label{subsec:additionalproperties}

Following the works mentioned above, we use the definition for a PUF family, introduced in \cite{PUFsUC}. A PUF family \puffamily{} is defined by two stateless probabilistic Turing machines \pufsample{} and \pufeval{}, which are not necessarily efficient.
The index sampling algorithm, \pufsample{}, outputs an index $\msf{id}$ and represents the manufacturing process of a PUF, while 
the evaluation algorithm, \pufeval{}, takes a challenge $s$ as input and outputs the corresponding response $\sigma$.

\begin{definition}
    Let $\puffamily = (\pufsample, \pufeval)$ be a pair of algorithms that run on security parameter $n$.
    We say that \puffamily{} is a $\prn{\rg, \dnoise}$-PUF family if it satisfies the following properties:
    \begin{itemize}
        \item \textbf{Index sampling:}
        The sampling algorithm $\pufsample_n$ outputs an index $\msf{id}$ from a fixed index set $\mathcal{I}_n$.
        Each $\msf{id} \in \mathcal{I}_n$ corresponds to a set of distributions $\mathcal{D}_\msf{id}$:
        for each challenge $s \in \bin^n$, $\mathcal{D}_\msf{id}(s)$ is a distribution on $\bin^{\rg(n)}$.

        \item \textbf{Evaluation:}
        For all challenges $s \in \bin^n$, the evaluation algorithm $\pufeval_n \prn{\msf{id}, s}$
        outputs a response $\sigma \in \bin^{\rg(n)}$ according to the distribution $\mathcal{D}_\msf{id}(s)$.

        \item \textbf{Bounded noise:}
        For all indices $\msf{id} \in \mathcal{I}_n$ and challenges $s \in \bin^n$,
        if $\sigma$ and $\sigma'$ are obtained from running $\pufeval_n \prn{\msf{id}, s}$ twice, then $\dham(\sigma, \sigma') < \dnoise(n)$,
        where \dham{} is the Hamming distance.
    \end{itemize}
    \label{def:PUFamily}
\end{definition}

The main security property of PUFs, unpredictability, is formalized in \cite{PUFsUC} using 
average min-entropy.
Let

\begin{itemize}
    \item $\PUF \gets \puffamily$ mean $\msf{id} \gets \pufsample_n$ and then restricting the name ``\PUF{}'' to this \msf{id};
    \item $\sigma \gets \PUF(s)$ mean $\sigma \gets \pufeval_n(\msf{id}, s)$;
    \item $\bm{\sigma}_{\bm{q}} \gets \PUF(\bm{q})$ mean $\bm{\sigma}_{\bm{q}} \gets \prn{\pufeval_n(\msf{id}, q)}_{q \in \bm{q}}$,
    where $\bm{q}$ is a list of queries.
\end{itemize}
\begin{definition}
    We say that a \prn{\rg, \dnoise}-PUF family $\mathcal{P} = (\pufsample, \pufeval)$ is \prn{\dmin, \entropybound}-unpredictable if
    for any $s \in \bin^n$ and list of queries $\bm{q}$ of polynomial size in $n$, the following condition holds:
    \[
        \dham(\bm{q}, s) \geq \dmin(n) \implies \condavgminentropy{\PUF(s)}{\PUF(\bm{q})} \geq \entropybound(n),
    \]
    where $\dham(\bm{q}, s) \coloneqq \min_i \dham(q_i, s)$.
    Such a PUF-family is called a \prn{\rg, \dnoise, \dmin, \entropybound}-PUF family.
\end{definition}

Since the PUF evaluation is inherently noisy, the fuzzy extractors (FE) from \cite{FuzzyExtractors} 
are used in \cite{PUFsUC} to convert noisy, high-entropy outcomes
of PUFs into reproducible random values.
\begin{definition}
    Let $\mathcal{M}$ be a metric space with distance function \msf{dis}.
    Consider a pair of efficient randomized algorithms $\fe = \prn{\Gen, \Rep}$ such that:
    \begin{itemize}
        \item \msf{Gen}, on input $w \in \mathcal{M}$, outputs a pair $(st, p)$, where $st \in \bin^l$ is called the secret string
        and $p \in \bin^*$ is called the helper data string;
        \footnote{
            Here, we use Kleene star notation, so $\bin^*$ denotes the set of all finite binary strings.
        }
        
        \item \msf{Rep}, on input an element $w \in \mathcal{M}$ and a helper data string $p \in \bin^*$, outputs a string $st$.
    \end{itemize}

    We say that \fe{} is an average-case \prn{\fem, \fel, \fet, \feeps}-FE if:
    \begin{itemize}
        \item \textbf{Correctness:}
        For all $w, w' \in \mathcal{M}$, if $\msf{dis} \prn{w, w'} \leq t$ and $(st, p) \gets \Gen(w)$, then $\Rep(w', p) = st$;
        
        \item \textbf{Security:}
        Consider some random variables $W$, which takes values $w \in \mathcal{M}$, and \EXTRA{}, which corresponds to some additional information
        related to $W$.
        Let $D_0$ and $D_1$ denote the distributions of the outputs from the programs depicted in Fig. \ref{fig:fe security}.
        Then,
        \[
            \condavgminentropy{W}{\EXTRA} \geq \entropybound \implies \SD \prn{ D_0, D_1 } < \feeps.
            \footnote{
                \SD{} denotes the statistical distance.
            }
        \]
        \begin{figure}
            \begin{pchstack}[center, space = 1cm]
                \procedure{$D_0$}{
                    (w, \extra) \gets (W, \EXTRA) \\
                    (st, p) \gets \Gen(w) \\
                    \pcreturn (st, p, \extra)
                }
    
                \procedure{$D_1$}{
                    (w, \extra) \gets (W, \EXTRA) \\
                    (st, p) \gets \Gen(w) \\
                    u \sample \bin^\fel \\
                    \pcreturn (u, p, \extra)
                }
            \end{pchstack}
            \caption{Programs for the security property of FEs.}
            \label{fig:fe security}
        \end{figure}
    \end{itemize}
    \label{def:fuzzyextractor}
\end{definition}
We refer to average-case FEs simply as FEs.
Now, consider the functions \fem{}, \fel{}, \fet{} and \feeps{} and algorithms \prn{\Gen, \Rep} that run on the security parameter $n$.
We say that $\fe = \prn{\Gen, \Rep}$ is a \prn{\fem, \fel, \fet, \feeps}-FE family if,
for each $n \in \N$, $\fe_n = \prn{\Gen_n, \Rep_n}$ is a \prn{\fem(n), \fel(n), \fet(n), \feeps(n)}-FE.
The following definition specifies the parameters for a FE family to be matching with a given PUF family.  
\begin{definition}
    Let \puffamily{} be a \prn{\rg, \dnoise, \dmin, \entropybound}-PUF family and \fe{} be a
    \prn{\fem_\fe, \fel_\fe, \fet_\fe, \feeps_\fe}-FE family.
    We say that \puffamily{} and \fe{} have matching parameters if:
    \begin{itemize}
        \item $\fe_n$ is defined on the metric space $\bin^{\rg(n)}$ with Hamming distance \dham{};
        \item $\fem_\fe(n) = \entropybound(n)$;
        \item $\fet_\fe(n) = \dnoise(n)$;
        \item $\feeps_\fe(n) = \abs{ \varepsilon(n) }$.
        \footnote{
            Here, $\varepsilon(n)$ denotes an arbitrary negligible function.
            We use this notation throughout out work.
        }
    \end{itemize}
\end{definition}

The following properties hold for PUF families and FE families with matching parameters:
    
\textbf{Response consistency:}
    Let $\sigma, \sigma'$ be responses of \PUF{} when queried twice with the same challenge $s$.
    If $(st, p) \gets \msf{Gen}(\sigma)$, then $\msf{Rep}(\sigma', p) = st$.

\textbf{Almost-uniformity:}
    Let $s \in \bin^n$ and consider the programs depicted in Fig. \ref{fig:almost unif}.
    Then, for each $n \in \N$, 
    \begin{align*}
        &\condavgminentropy{\PUF(s)}{\EXTRA} \geq \entropybound(n)\\ 
        &\implies \SD \prn{ D_0, D_1 } < \feeps_\fe(n).
    \end{align*}
    \begin{figure}[ht]
        \begin{pchstack}[center, space = 1cm]
            \procedure{$D_0$}{
                \PUF \gets \puffamily \\
                \sigma \gets \PUF(s) \\
                (st, p) \gets \Gen(\sigma) \\
                \extra \gets \EXTRA \\
                \pcreturn (st, p, \extra)
            }
    
            \procedure{$D_1$}{
                \PUF \gets \puffamily \\
                \sigma \gets \PUF(s) \\
                (st, p) \gets \Gen(\sigma) \\
                u \sample \bin^{\rg(n)} \\
                \extra \gets \EXTRA \\
                \pcreturn (u, p, \extra)
            }
        \end{pchstack}
        \caption{Programs for the almost-uniformity property.}
        \label{fig:almost unif}
    \end{figure}

Notice that if we take the random variable $\EXTRA = \PUF(\bm{q})$ with $\dham(\bm{q}, s) \geq \dmin(n)$, 
unpredictability ensures $\condavgminentropy{\PUF(s)}{\EXTRA} \geq \entropybound(n)$ and thus 
almost-uniformity ensures $\SD \prn{ D_0, D_1 } < \feeps_\fe(n)$.
This particular case, known as \textbf{extraction independence}, is the approach used in 
\cite{PUFsUC}.
However, we adopt this more general property, as it offers greater flexibility in security proofs.

Furthemore, an additional property called \textbf{well-spread domain} was also proven in 
\cite{PUFsUC}.
Informally, it states that if an honest party generates a challenge $s$ uniformly at random,
then an adversary attempting to choose a challenge close to $s$ (that is, within a distance 
smaller than \dmin{}) should only have a negligible probability of success \cite{PUFsUC};
this was achieved by assuming $\dmin(n) \in o(n / \log(n))$.
The goal was to use it along with extraction independence to arrive at
what we call the \textbf{indistinguishability property}: if an honest party generates a challenge 
$s$ uniformly at random, then the output of the FE applied to the response $\PUF(s)$ should be 
indistinguishable from a uniformly random response, even if the adversary has access to \PUF{} 
itself.

However, ensuring that this still holds when the adversary has access to $\PUF(s)$
-- we call this the \textbf{close query (CQ) property} --
is essential for indistiguishability, and it turns out to be non-trivial, as $\PUF(s)$ may 
inadvertently reveal some information about $s$.
Since this had not been considered in previous works, instead of $\dmin(n) \in o(n / \log(n))$ 
we had to make the following stronger assumption:

\textbf{Preimage entropy:}
    Let \puffamily{} be a PUF family and consider the neighborhood $B_n^d(x) = \set{ y \in \bin^n : \dham(x, y) < d }$ around $x \in \bin^n$.
    \footnote{Notice that the size of these neighborhoods does not depend on $x$.
    More specifically, $\abs{B_n^d} = \sum_{k = 0}^{d - 1} \binom{n}{k}$.}
    Then,
        \[
            \abs{B_n^{\dmin(n)}} 2^{-\condavgminentropy{S}{\PUF(S)}} = \varepsilon(n).
        \]

In Appendix \ref{appendix:wellspread}, we formally present the CQ property and prove that a 
PUF family satisfying the preimage entropy property also satisfies the CQ property.
In Appendix \ref{appendix:indist}, we formally present the indistinguishability property and prove that 
it follows from the CQ property.

The next property we need is what we call the \textbf{challenge-response pair (CRP) guessing} property, 
which informally states that it should be hard to generate a valid CRP of a PUF, without querying the 
PUF ``close'' to the corresponding challenge.
\footnote{
    More precisely, this refers to a tuple $(s,st,p)$, derived from an actual CRP $(s,\sigma)$ using a FE. 
}
In this work, we assume that the CRP guessing property holds, even though we would ultimately like to 
eliminate it by reducing it to some of the other PUF properties.
However, it seems that such reductions require to use specific descriptions for the FE.
We opted to assume that the CRP property holds instead of restricting to specific FE descriptions.
A short discussion about these possible reductions, as well as the formal statement of the CRP guessing 
property can be found in Appendix \ref{appendix:crp}.

Finally, we also need what we call the \textbf{test query} property. 
Suppose an honest party sends its PUF to an adversary who then returns it.
How can the honest party be sure that the PUF it received is indeed the one it originally created?
In \cite{UCComm}, it was noted that the honest party could query the PUF on a randomly selected 
challenge (a \textit{test query}) and verify the response upon receiving the PUF back.
If the returned PUF passes this test, then, with overwhelming probability, it is the original PUF.
However, a formal proof of this property was not provided.
Its validity depends on the model of malicious PUFs considered.
For example, in \cite{BoundedStateOT}, where an adversary can construct a
malicious PUF that encapsulates a PUF received from a different party, the adversary could create a PUF 
that differs from the original one only on a specific challenge.
Clearly, the test query property would fail in that scenario, which highlights the need to prove it 
with caution.
In Appendix \ref{appendix:testquery}, we present the formal statement of the test query property, as well as the 
proof that it follows from the preimage entropy and CRP properties.

\subsection{Attack on the ideal extractable commitment from \cite{UCComm}}
\label{subsec:issues}

Let us first present some definitions from \cite{UCComm}.
\begin{definition}
    A \textup{\textbf{commitment scheme in the \FComMPUF{}-hybrid model}} is a tuple of PPT algorithms $\com = (\S, \R)$
    that run on security parameter $n$ and have oracle access to \FComMPUF{}, implementing the following functionality:
    \begin{itemize}
        \item \textup{\textbf{Inputs:}}
        \S{} receives as input a string $x \in \bin^k$.
        
        \item \textup{\textbf{Commitment phase:}}
        \S{} interacts with \R{} to commit to the string $x$; we denote this by $\textup{\texttt{Commit}}^\com (x)$.

        \item \textup{\textbf{Decommitment phase:}}
        \S{} sends $x$ and some decommitment data to \R{}, which outputs either $x$, if it accepts the decommitment, or $\bot$, otherwise;
        we denote this by $\textup{\texttt{Open}}^\com (x)$.
    \end{itemize} 
\end{definition}

Notice that, in the definition above, we allow commitments to strings of any length $k$, rather than 
restricting them to bit commitments as was done in \cite{UCComm}.
\begin{definition}
    A commitment scheme $\com = (\S, \R)$ is an \textup{\textbf{ideal commitment scheme in the \FComMPUF{}-hybrid model}}
    if it satisfies the following properties:
    \begin{itemize}
        \item \textup{\textbf{Completeness:}}
        If \S{} and \R{} follow their prescribed strategy, then \R{} accepts the decommitment with probability 1.

        \item \textup{\textbf{Computationally Hiding:}}
        Let $x^0$ and $x^1$ be different strings in $\bin^k$.
        Consider the interaction between an honest sender \S{} and a malicious receiver \Radv{} depicted in Fig. \ref{fig:single hiding}.
        \footnote{
            Here, \S{} acts as a challenger and \Radv{} as a distinguisher.
        }
        \begin{figure}[ht]
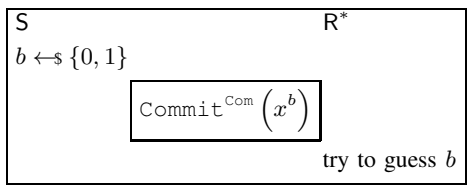

            \textup{
                \procb{}{
                    \S \< \< \Radv \\
                    b \sample \bin \< \< \\
                    \< \pcbox{ \texttt{Commit}^\com \prn{ x^b } } \< \\
                    \< \< \text{try to guess } b
                }
            }
            \caption{Hiding interaction.}
            \label{fig:single hiding}
        \end{figure}

        We say that \com{} is \textbf{computationally hiding} if for all PPT malicious receivers \Radv{},
        \textup{
            \[
                \prob{ \Radv \prn{ \textup{\texttt{Commit}}^\com \prn{ x^B } } = B } = \frac{1}{2} + \varepsilon(n).
            \]
        }
        
        \item \textup{\textbf{Statistically Binding:}}
        Consider the interaction between a malicious sender \Sadv{} and an honest receiver \R{} depicted in Fig. \ref{fig:single binding},
        where \texttt{Commit} denotes the commitment phase of \com{}, in which \Sadv{} may behave in a malicious way.
        Furthermore, $d_1$ and $d_2$ denote two sequences of actions that lead to decommitments.
        \footnote{
            This includes the case where \Sadv{} sends different messages to some PUF it created, depending on the string
            it is going to decommit to.
        }
        We say that \Sadv{} is successful when $d_1$ and $d_2$ lead to successful decommitments to different strings.
        \begin{figure}[ht]
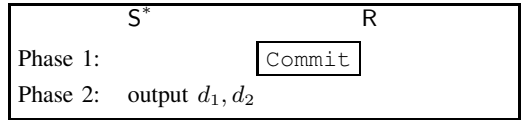

            \textup{
                \procb{}{
                    \< \Sadv \< \< \R \hspace{50pt} \\
                    \text{Phase 1:} \hspace{10pt} \< \< \pcbox{ \texttt{Commit} } \< \\
                    \text{Phase 2:} \hspace{10pt} \< \text{output } d_1, d_2 \< \<
                }
            }
            \caption{Binding interaction.}
            \label{fig:single binding}
        \end{figure}

        We say that \com{} is \textbf{statistically binding} if all malicious senders \Sadv{} in the interaction depicted
        in Fig. \ref{fig:single binding} succeed with negligible probability.  
    \end{itemize}
\end{definition}

In the definition above, we adopted a different (but equivalent) formulation for hiding, where a 
malicious receiver \Radv{} acts as a distinguisher attempting to guess the committed string.
We found that this provides a clearer framework for proving our results.

The protocols constructed in \cite{UCComm} were described as ideal, 
meaning both statistical hiding and statistical binding.
However, they were only statistically hiding under the assumption that the adversary makes a polynomial 
number of queries to \FComMPUF{}.
Rather than relying on this assumption, we instead restrict our analysis to PPT adversaries, who are 
inherently limited to making a polynomial number of queries.  
However, for simplicity, we continue to refer to them as ideal.
\begin{definition}
    An algorithm $M$ has \textup{\textbf{interface access to the functionality \FComMPUF{}}} 
    with respect to a protocol in the \FComMPUF{}-hybrid model if $M$ has oracle access to \FComMPUF{} and can observe any query
    made by any party to honest PUFs during the protocol execution.
\end{definition}
\begin{definition}
    \label{def:single extractability}
    A commitment scheme $\com = (\S, \R)$ is an \textup{\textbf{ideal extractable commitment scheme in the \FComMPUF{}-hybrid model}}
    if \com{} is an ideal commitment and there exists a PPT extractor \ext{}
    having interface access to \FComMPUF{} such that, for all malicious senders \Sadv{}, it interacts with \Sadv{} as depicted in Fig. \ref{fig:single extractability}
    and satisfies the following properties:
    \begin{itemize}
        \item \textup{\textbf{Simulation:}}
        The view of \Sadv{}
        \footnote{
            That is, the distribution of the messages exchanged during the interaction, as well as \Sadv{}'s private
            information.
        }
        when interacting with \ext{} is identical to the view when interacting with an honest receiver \R{}.

        \item \textup{\textbf{Extraction:}}
        \Sadv{} only decommits successfully to some string that is different from what \ext{} outputs
        \footnote{
            Notice that if \ext{} outputs $\bot$, this means that \Sadv{} cannot decommit successfully to any string.
        }
        with negligible probability, that is,
        \[
            \prob{\Sadv{} \text{ decommits successfully to } X \neq X^*} = \varepsilon(n).
        \]
    \end{itemize}
    \begin{figure}[ht]
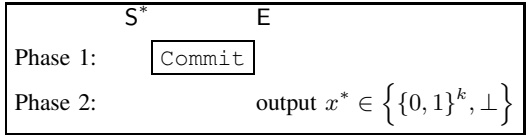

        \textup{
            \procb{}{
                \< \Sadv{} \< \< \ext \hspace{50pt} \\
                \text{Phase 1:} \hspace{10pt} \< \< \pcbox{ \texttt{Commit} } \< \\
                \text{Phase 2:} \hspace{10pt} \< \< \< \text{output } x^* \in \set{\bin^k, \bot}
            }
        }
        \caption{Extractability interaction.}
        \label{fig:single extractability}
    \end{figure}    
\end{definition}

In \cite{UCComm}, the ideal extractable commitment \ExtPUF{} (depicted in Fig. \ref{fig:original 
extpuf}) was constructed from the ideal commitment \CPUF{} from \cite{MalPUFs} (depicted in Fig. 
\ref{fig:single nocomm cpuf} in Appendix \ref{appendix:newproofs}), which is essentially based on the one from \cite{Naor}.
We found, however, some issues in this construction.

The first issue concerns the proof for \CPUF{} presented in \cite{MalPUFs}.
Specifically, the hiding argument does not fully address the potential of an adversary learning about 
$s$ after receiving information that depends on $\PUF(s)$.
This is precisely where the aforementioned CQ property comes into play.
Moreover, the proof for binding appears to be incomplete.
To address this, we adapt the original proof from \cite{Naor}.
Our hiding argument relies on the indistinguishability property of PUFs, while our binding 
argument generalizes the proof to account for PUFs' inherent noise by incorporating certain 
entropy properties.
Our proof can be found in Appendix \ref{appendix:newproofs} (Theorem \ref{single nocomm cpuf}).

The second, and most critical, issue has to do with \ExtPUF{} protocol and its corresponding
extractor (see Fig. \ref{fig:original extpuf} and \ref{fig:original extractor}).
This protocol uses the following parameters:
\begin{itemize}
    \item a PUF family $\puffamily_\E$ and a fuzzy extractor family $\fe_\E = (\Gen_\E, \Rep_\E)$
    \footnote{
        For simplicity, however, in the protocol description we omit the notation specific fuzzy extractors and write \Gen{} and \Rep{}
        generically, assuming the appropriate fuzzy extractor is used with each PUF.
        This convention will be followed in subsequent protocols as well.
    }
    with matching parameters;
    \item a family $(\Enc, \Dec)$ of $(kl, L, 2\prn{\dmin}_\E - 1)$-error-correcting codes
    \footnote{
        See Definition \ref{def:ecc} in Appendix \ref{appendix:newproofs}.
    }
    for some $L$, with $l(n) = 3n$;
    \item a PUF family $\puffamily_\CS$ and a fuzzy extractor family $\fe_\CS = (\Gen_\CS, \Rep_\CS)$ with matching parameters
    such that $\fel_{\fe_\CS}(n) = kl$;
    \item a PUF family $\puffamily_\CR$ and a fuzzy extractor family $\fe_\CR = (\Gen_\CR, \Rep_\CR)$ with matching parameters
    such that $\fel_{\fe_\CR}(n) = ml$, with $m = \left| st_\E \mathbin \Vert p_\E \right|$.
\end{itemize}
Furthermore, in the protocol description \TQ{} denotes a test query.

We are now ready to describe an attack on this protocol.
Consider a malicious sender \Sadv{} that behaves just like an honest \S{} committing to the 
bit 0, except that it also queries \PUFE{} on $\Enc \prn{ st_\CS \oplus r_\CS }$.
Then, $c_\CS = st_\CS$ and
$\mathcal{Q} = \set{\Enc \prn{st_\CS}, \Enc \prn{ st_\CS \oplus r_\CS }}$, which means
\begin{itemize}
    \item for $q = \Enc \prn{st_\CS}$, we have $\msf{Dec}(q) \oplus \prn{ 0^l \land r_\CS } = st_\CS = c_\CS$ and so 0 is extracted;
    \item for $q = \Enc \prn{ st_\CS \oplus r_\CS }$, we have $\msf{Dec}(q) \oplus \prn{ 1^l \land r_\CS } = st_\CS \oplus r_\CS \oplus r_\CS = c_\CS$
    and so 1 is extracted.
\end{itemize}
Therefore, in this case \ext{} always outputs $\bot$.
However, \Sadv{} can always decommit successfully to 0, breaking the extraction property.

The original goal in \cite{UCComm} behind \ExtPUF{} was to base its extractability on the 
binding property of \CPUF{}.
Indeed, it was meant to force \Sadv{} to query \PUFE{} on an opening of the commitment 
$\CPUF(x)$, and thus binding it to $x$.
However, as we have seen, $st_\CS$ is not an opening of that commitment.

To prevent this attack, our approach will be to adjust the protocol so that \S{}
returns \PUFE{} before \R{} sends $r$, as depicted in Fig. \ref{fig:single nocomm extpuf}.
An immediate consequence of this change is that \PUFE{} must be stateless; 
otherwise \Radv{} could learn information about the string being committed.

In this modified setup, \S{} no longer needs to commit to $st_\E \mathbin \Vert p_\E$.
Indeed, we do not need to worry about the malicious senders only querying \PUFE{} in the 
decommitment phase, since they must return \PUFE{} during the commitment phase.
Moreover, the length of the string $r$ can be reduced to $kn$ instead of $3kn$, as the 
protocol no longer depends on the statistical binding of \CPUF{}.
To see why, notice that in \CPUF{}, a malicious sender could indirectly communicate 
information about $r$ to \PUF{} by selecting an appropriate $s$ in the decommitment phase, 
which is what led us to extend the size of $r$ in the first place.
However, in this situation, the sender is required to query \PUFE{} with the response from 
\PUF{} before receiving $r$, so this is not a problem anymore.
In summary, the protocol parameters are now the following:
\begin{itemize}
    \item a PUF family $\puffamily_\E$ and a fuzzy extractor family $\fe_\E = (\Gen_\E, \Rep_\E)$ with matching parameters;
    \item a family $(\Enc, \Dec)$ of $\prn{ kn, L, 2\prn{\dmin}_\E - 1 }$-error-correcting codes for some $L$;
    \item a PUF family $\puffamily$ and a fuzzy extractor family $\fe = (\Gen, \Rep)$ with matching parameters
    such that $\fel_\fe(n) = kn$.
\end{itemize}

Finally, the original protocol can also be simplified by noticing that we can minimize the 
number of PUF exchange phases by sending \PUF{} and \PUFE{} simultaneously.
This is a first step toward achieving a more efficient UC-secure commitment protocol.
Our proof can be found in Appendix \ref{appendix:newproofs} (Theorem \ref{single nocomm extpuf}).

\begin{figure*}[ht]
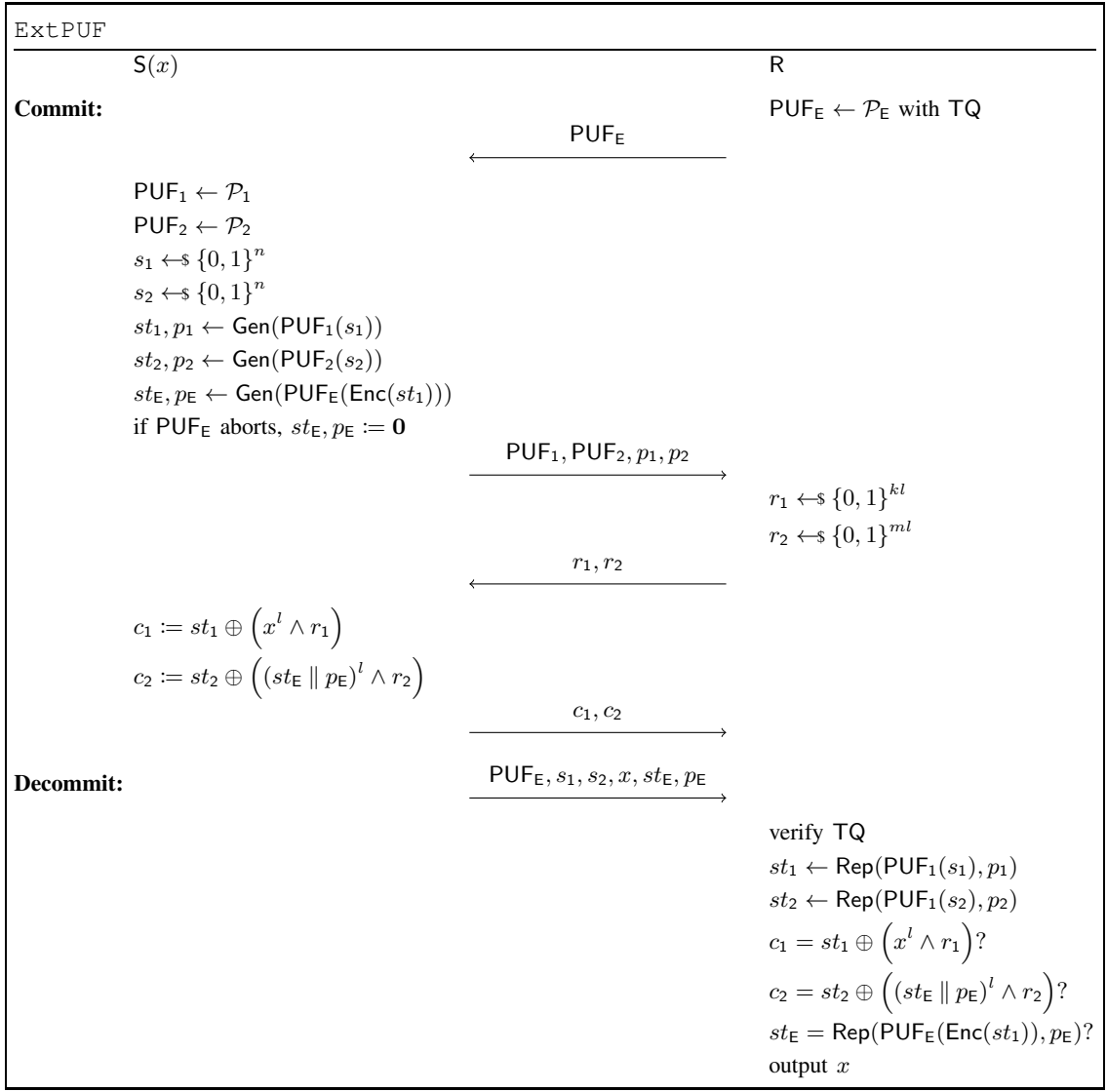

    \procb{\ExtPUF}{
        \< \S(x) \< \< \R \< \\[1ex]
        \textbf{Commit:   } \< \< \< \PUFE \gets \puffamily_\E \text{ with } \TQ \\[-1.5ex]
        \< \< \sendmessageleft*{\PUFE} \< \\
        \< \PUFCS \gets \puffamily_\CS \< \< \\
        \< \PUFCR \gets \puffamily_\CR \< \< \\
        \< s_\CS \sample \bin^n \< \< \\
        \< s_\CR \sample \bin^n \< \< \\
        \< st_\CS, p_\CS \gets \Gen (\PUFCS(s_\CS)) \< \< \\
        \< st_\CR, p_\CR \gets \Gen (\PUFCR(s_\CR)) \< \< \\
        \< st_\E, p_\E \gets \Gen (\PUFE(\Enc(st_\CS))) \< \< \\
        \< \text{if \PUFE{} aborts, } st_\E, p_\E \coloneqq \bm{0} \< \< \\[-1.5ex]
        \< \< \sendmessageright*{\PUFCS, \PUFCR, p_\CS, p_\CR} \< \\[-1.5ex]
        \< \< \< r_\CS \sample \bin^{kl} \\
        \< \< \< r_\CR \sample \bin^{ml} \\[-1.5ex]
        \< \< \sendmessageleft*{r_\CS, r_\CR} \< \\
        \< c_\CS \coloneq st_\CS \oplus \prn{ x^l \land r_\CS } \< \< \\
        \< c_\CR \coloneq st_\CR \oplus \prn{ ( st_\E \mathbin \Vert p_\E )^l \land r_\CR } \< \< \\[-1.5ex]
        \< \< \sendmessageright*{c_\CS, c_\CR} \< \\
        \textbf{Decommit:   } \< \< \sendmessageright*{\PUFE, s_\CS, s_\CR, x, st_\E, p_\E} \< \\
        \< \< \< \text{verify } \TQ \\
        \< \< \< st_\CS \gets \Rep (\PUFCS(s_\CS), p_\CS) \\
        \< \< \< st_\CR \gets \Rep (\PUFCS(s_\CR), p_\CR) \\
        \< \< \< c_\CS = st_\CS \oplus \left( x^l \land r_\CS \right) ?\\
        \< \< \< c_\CR = st_\CR \oplus \left( ( st_\E \mathbin \Vert p_\E )^l \land r_\CR \right) ?\\
        \< \< \< st_\E = \Rep (\PUFE( \Enc (st_\CS)), p_\E) ? \\
        \< \< \< \text{output } x
    }
    \caption{The original \ExtPUF{} protocol.}
    \label{fig:original extpuf}
\end{figure*}

\begin{figure}
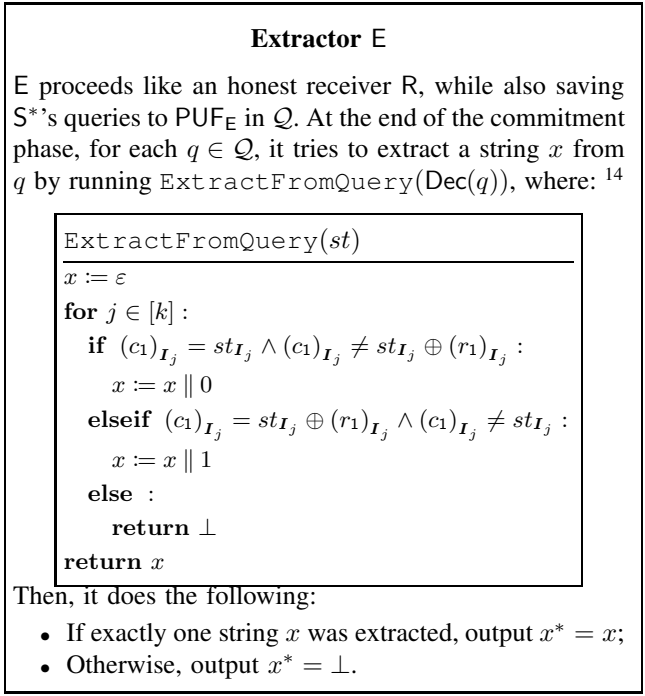

    \begin{center}
        \noindent\fbox{%
            \parbox{0.95\linewidth}{
                \begin{center}
                    \textbf{Extractor \ext}
                \end{center}
        
                \ext{} proceeds like an honest receiver \R{}, while also saving \Sadv{}'s queries to \PUFE{} in $\mathcal{Q}$.
                At the end of the commitment phase, for each $q \in \mathcal{Q}$, it tries to extract a string $x$ from $q$ by running
                \texttt{ExtractFromQuery}$(\msf{Dec}(q))$, where:
                \footnote{
                    In this procedure, $\varepsilon$ denotes the empty string.
                }              
                \pcb[head = \texttt{ExtractFromQuery}$(st)$]{
                    x \coloneqq \varepsilon \\
                    \pcfor j \in [k]: \\
                    \pcind \pcif \prn{c_\CS}_{\bm{I}_j} = st_{\bm{I}_j} \land \prn{c_\CS}_{\bm{I}_j} \neq st_{\bm{I}_j} \oplus \prn{r_\CS}_{\bm{I}_j}: \\
                    \pcind \pcind x \coloneq x \mathbin \Vert 0 \\
                    \pcind \pcelseif \prn{c_\CS}_{\bm{I}_j} = st_{\bm{I}_j} \oplus \prn{r_\CS}_{\bm{I}_j} \land \prn{c_\CS}_{\bm{I}_j} \neq st_{\bm{I}_j}: \\
                    \pcind \pcind x \coloneq x \mathbin \Vert 1 \\
                    \pcind \pcelse : \\
                    \pcind \pcind \pcreturn \bot \\
                    \pcreturn x
                }
    
                Then, it does the following:
                \begin{itemize}
                    \item If exactly one string $x$ was extracted, output $x^* = x$;
                    \item Otherwise, output $x^* = \bot$.
                \end{itemize}
            }
        }%
    \end{center}
    \caption{The original extractor defined in \cite{UCComm}.}
    \label{fig:original extractor}
\end{figure}

\begin{figure*}
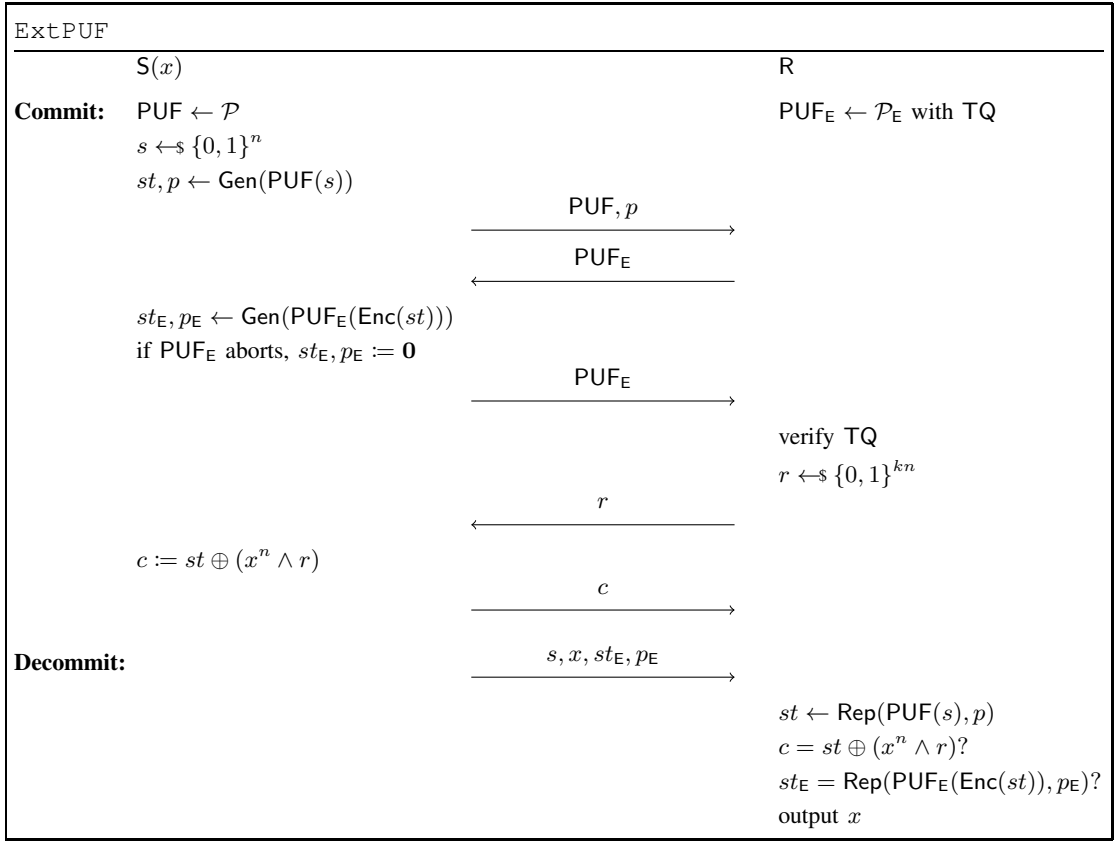

    \procb{\ExtPUF}{
        \< \S(x) \< \< \R \< \\[1ex]
        \textbf{Commit:   } \< \PUF \gets \puffamily \< \< \PUFE \gets \puffamily_\E \text{ with } \TQ \\
        \< s \sample \bin^n \< \< \\
        \< st, p \gets \Gen (\PUF(s)) \< \< \\[-1.5ex]
        \< \< \sendmessageright*{\PUF, p} \< \\[-1.5ex]
        \< \< \sendmessageleft*{\PUFE} \< \\
        \< st_\E, p_\E \gets \Gen (\PUFE(\Enc(st))) \< \< \\
        \< \text{if \PUFE{} aborts, } st_\E, p_\E \coloneqq \bm{0} \< \< \\[-1.5ex]
        \< \< \sendmessageright*{\PUFE} \< \\
        \< \< \< \text{verify } \TQ \\
        \< \< \< r \sample \bin^{kn} \\[-1.5ex]
        \< \< \sendmessageleft*{r} \< \\
        \< c \coloneq st \oplus \prn{ x^n \land r } \< \< \\[-1.5ex]
        \< \< \sendmessageright*{c} \< \\
        \textbf{Decommit:   } \< \< \sendmessageright*{s, x, st_\E, p_\E} \< \\
        \< \< \< st \gets \Rep (\PUF(s), p) \\
        \< \< \< c = st \oplus \left( x^n \land r \right) ?\\
        \< \< \< st_\E = \Rep (\PUFE( \Enc (st)), p_\E) ? \\
        \< \< \< \text{output } x
    }
    \caption{The modified \ExtPUF{} protocol.}
    \label{fig:single nocomm extpuf}
\end{figure*}
\section{the compiler from \cite{UCComm}}
\label{sec:fixcompiler}
\subsection{Collective commitments}

As mentioned in Section \ref{subsec:contibutions}, we generalized the definition of commitments to accommodate the commitment of many strings at once. This yields what we call a \emph{collective commitment scheme}, and the corresponding syntax, where
$N(n)$ denotes the number of strings being committed and $k(n)$ their size, is given in the following definition:
\begin{definition}
    A \textup{\textbf{collective commitment scheme in the \FComMPUF{}-hybrid model}} is a tuple of PPT algorithms $\collcom = (\S, \R)$
    that run on security parameter $n$ and have oracle access to \FComMPUF{}, implementing the following functionality:
    \begin{itemize}
        \item \textup{\textbf{Inputs:}}
        \S{} receives as inputs strings $x^1, \cdots, x^{N(n)} \in \bin^{k(n)}$.

        \item \textup{\textbf{Commitment phase:}}
        \S{} commits to $\bm{x} = \prn{x^1, \cdots, x^{N(n)}}$, which we denote by $\textup{\texttt{Commit}}^\collcom (\bm{x})$.

        \item \textup{\textbf{Decommitment of commitments in a set $I \subseteq [N]= \set{1, \cdots, N}$:}}
        
        \S{} sends $\set{ \prn{ i, x^i } }_{i \in I}$ and some decommitment data to \R{}, which outputs either $\set{ \prn{ i, x^i } }_{i \in I}$,
        if it accepts the decommitment, or $\bot$, otherwise.
        We denote this by $\textup{\texttt{Open}}^\collcom \prn{ \prn{x^i}_{i \in I} }$, and we refer to this phase as the \textit{decommitment of $I$}.
    \end{itemize}

    When using this protocol, there can be many decommitment phases, not necessarily at the same time.
\end{definition}

In the following, consider a fixed collective commitment scheme in the \FComMPUF{}-hybrid model, where
we only consider functions $N(n)$ and $k(n)$ that are polynomial in $n$.

Now we need to define what it means for such a protocol to be hiding.
The intuitive idea is that the committed strings remain hidden until they are revealed — even if other strings have already been opened.
This property must hold even within a more complex interaction.

Of course, this interaction must be restricted in certain ways.
For instance, if $x$ is one of the strings committed by the sender \S{} and \S{} later sends $x$ to a malicious receiver \Radv{},
it would no longer remain hidden, even without explicitly opening the commitment.
Therefore, \S{} must be restricted from sending any messages that depend on the strings intended to remain hidden throughout the interaction.
Additionally, \S{} must be restricted from sharing any information generated during the commitment phase, as this could aid \Radv{}
in learning about the committed strings.
Thus, \S{} should only interact with the commitment scheme as a black box, ensuring no internal details are leaked.
Furthermore, the interaction may involve \S{} not knowing which strings it will commit to at the start.
However, there must be a clear point where \S{} defines the strings and decides which ones will eventually be revealed.
We formalize this in the following definition:
\begin{definition}
    Consider the interaction between an honest sender \S{} and a malicious receiver \Radv{} depicted in Fig. \ref{fig:collective hiding},
    where $\ensemble{\Omega_n}$ is a collection of finite sets and $\INTER^\collcom$ denotes an interaction such that:
    \begin{itemize}
        \item \S{} and \Radv{} can interact arbitrarily, as long as the messages sent by \S{} do not depend on $w$;
        \item There is a moment in the interaction where \S{} defines a function $\STRINGS_n : \Omega_n \to \prn{ \bin^k }^{N(n)}$
        and sets $\OPEN_n \subseteq \Const \prn{ \STRINGS_n }$
        \footnote{
            For a function $f = \prn{ f^1, \cdots, f^N } : X \to Y^N$, we define $\Const(f) = \set{ i \in [N] : f^i \text{ is constant} }$.
        }
        and $\CLOSED_n \coloneqq [N] \setminus \OPEN_n$;
        \item After that, \S{} commits to $\STRINGS_n(w)$ using the protocol \collcom{} as a black-box.
        This is the only time in the interaction where \S{} has access to $w$. 
        Furthermore, since \S{} runs the commitment as a black-box, it does not have access to the information of the commitment outside
        the commitment and decommitment interactions.
        \footnote{
            Of course the same cannot be said about \Radv{}, since it is malicious.
        }
        \item \S{} decommits the strings in $\OPEN_n$ (not necessarily at the same time).
    \end{itemize}
    \begin{figure}[ht]
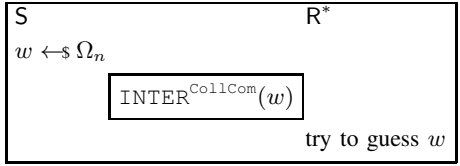

        \procb{}{
            \S \< \< \Radv \\
            w \sample \Omega_n \< \< \\
            \< \pcbox{ \INTER^\collcom (w) } \< \\
            \< \< \text{try to guess } w
        }
        \caption{Collective hiding interaction.}
        \label{fig:collective hiding}
    \end{figure} 

    We say that \collcom{} is \textbf{computationally hiding} if all interactions \INTER{} and PPT malicious receivers \Radv{}
    in the interaction depicted in Fig. \ref{fig:collective hiding} satisfy
    \[
        \prob{ \Radv \prn{ \INTER^\collcom (W) } = W } = \frac{1}{|\Omega_n|} + \varepsilon(n).
    \]
\end{definition}

Following the same idea, we also need to define the binding property within a more complex interaction.
As with the hiding definition, we ensure that \R{} interacts with the commitment scheme as a black box,
preventing it from sending any information that could aid the malicious sender \Sadv{}.
\begin{definition}
    Consider the interaction between a malicious sender \Sadv{} and an honest receiver \R{} depicted in Fig. \ref{fig:collective binding}, where    
    \INTER{} denotes an interaction where \Sadv{} makes a possibly malicious commitment using \collcom{}.
    Throughout this interaction, \R{} uses \collcom{} honestly as a black-box.
    Furthermore, $d_1$ and $d_2$ denote two sequences of actions that lead to decommitments.
    \footnote{
        This can include some decommitments.
    }
    We say that \Sadv{} is successful when $d_1$ and $d_2$ lead to successful decommitments of some
    $i \in [N]$ to different strings.
    \begin{figure}[ht]
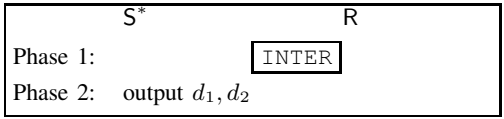

        \textup{
            \procb{}{
                \< \Sadv \< \< \R \hspace{50pt} \\
                \text{Phase 1:} \hspace{10pt} \< \< \pcbox{ \INTER } \< \\
                \text{Phase 2:} \hspace{10pt} \< \text{output } d_1, d_2 \< \<
            }
        }
        \caption{Collective binding interaction.}
        \label{fig:collective binding}
    \end{figure}
    
    We say that \collcom{} is \textbf{statistically binding} if for all interactions \INTER{}, all malicious senders \Sadv{}
    succeed with negligible probability.
\end{definition}

Finally, the same happens for extractability:
\begin{definition}
    Consider the interaction between a malicious sender \Sadv{} and some extractor \ext{} depicted in Fig. \ref{fig:collective extractability},   
    where \INTER{} denotes an interaction in which \Sadv{} makes a possibly malicious commitment using \collcom{}.
    We say that \collcom{} is \textbf{extractable} if there exists a PPT extractor \ext{}
    having interface access to \FComMPUF{} such that all interactions \INTER{} and malicious committers \Sadv{} satisfy the following properties:
    \begin{itemize}
        \item \textup{\textbf{Simulation:}}
        The view of \Sadv{} when interacting with \ext{} is identical to the view when interacting with an honest receiver \R{}.

        \item \textup{\textbf{Extraction:}}
        \Sadv{} only decommits successfully to some string that is different from what \ext{} outputs
        with negligible probability, that is,
        \begin{align*}
            &\prob{ \Sadv{} \text{ decommits some $i$ successfully to } X^i \neq \prn{X^*}^i } \\
            &= \ \varepsilon(n).
        \end{align*}
    \end{itemize}
    \begin{figure}[ht]
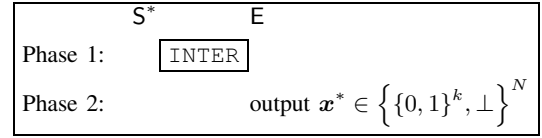

        \procb{}{
            \< \Sadv{} \< \< \ext \hspace{50pt} \\
            \text{Phase 1:} \hspace{10pt} \< \< \pcbox{ \INTER } \< \\
            \text{Phase 2:} \hspace{10pt} \< \< \< \text{output } \bm{x}^* \in \set{\bin^k, \bot}^N
        }
        \caption{Collective extractability interaction.}
        \label{fig:collective extractability}
    \end{figure}
\end{definition}

Just as before, we can define the concepts of ideal and ideal extractable collective commitment schemes in the \FComMPUF{}-hybrid model analogously.

\subsection{Adapting the compiler}
As previously discussed, since we are working with multiple commitments, we need to use a 
collective commitment scheme \collcom{}.
This approach not only ensures the security of these commitments is maintained but also 
allows us to optimize the protocol.
Our revised version of \texttt{UCCompiler} is depicted in Fig. \ref{fig:revised uccompiler}.
\footnote{
We define $[n, 2] = \set{1, 3, \cdots, 2n + 1}$.
}
It uses the parallelized version of the \texttt{BlobEquality} protocol from \cite{UCComm}, which we 
call \texttt{BlobEqualities}, and is depicted in Fig. \ref{fig:blobequalities}. 
We prove that our revised version of \texttt{UCCompiler} applied on ideal extractable commitments \collcom{} results in UC-secure commitments,  in Appendix \ref{appendix:ucproof}, see Theorem \ref{thm:ucappendix}. 
\begin{figure*}[ht]
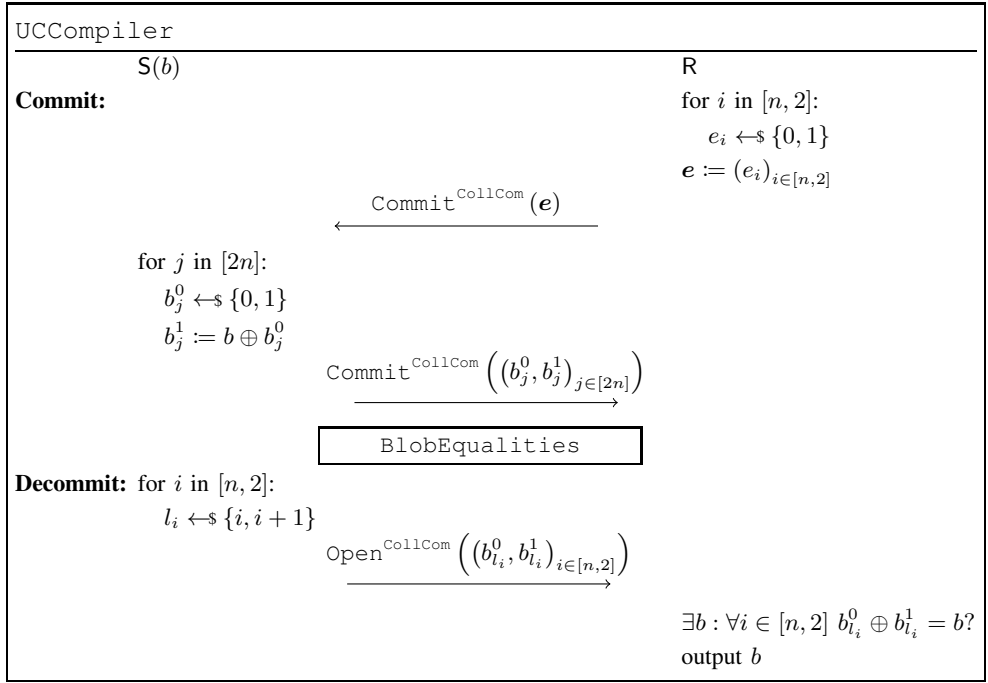

    \procb{\texttt{UCCompiler}}{
        \< \S(b) \< \< \R \\
        \textbf{Commit:   } \< \< \< \text{for $i$ in $[n, 2]$:} \\
        \< \< \< \hspace{10pt} e_i \sample \bin \\
        \< \< \< \bm{e} \coloneqq \prn{e_i}_{i \in [n, 2]} \\[-1.5ex]
        \< \< \sendmessageleft*{ \texttt{Commit}^\collcom \prn{ \bm{e} } } \< \\
        \< \text{for $j$ in $[2n]$:} \< \< \\
        \< \hspace{10pt} b_j^0 \sample \bin \< \< \\
        \< \hspace{10pt} b_j^1 \coloneqq b \oplus b_j^0 \< \< \\[-1.5ex]
        \< \< \sendmessageright*{ \texttt{Commit}^\collcom \prn{ \prn{b_j^0, b_j^1}_{j \in [2n]} } } \< \\
        \< \< \pcbox{\texttt{\hspace{20pt}BlobEqualities\hspace{20pt}}} \< \\
        \textbf{Decommit:   } \< \text{for $i$ in $[n, 2]$:} \< \< \\
        \< \hspace{10pt} l_i \sample \set{i, i + 1} \< \< \\[-1.5ex]
        \< \< \sendmessageright*{ \texttt{Open}^\collcom \prn{ \prn{ b_{l_i}^0, b_{l_i}^1 }_{i \in [n, 2]} } } \< \\
        \< \< \< \exists b : \forall i \in [n, 2] \ b_{l_i}^0 \oplus b_{l_i}^1 = b ? \\
        \< \< \< \text{output } b
    }
    \caption{The revised version of the \texttt{UCComm} protocol, which uses a collective commitment \collcom{}.}
    \label{fig:revised uccompiler}
\end{figure*}
\begin{figure*}[ht]
    \procb{\texttt{BlobEqualities}}{
        \S \< \< \R \\
        \text{for $i$ in $[n, 2]$:} \< \< \\
        \hspace{10pt} y_i \coloneqq b_i^0 \oplus b_{i + 1}^0 \< \< \\
        \bm{y} \coloneqq \prn{y_i}_{i \in [n, 2]} \< \< \\[-1.5ex]
        \< \sendmessageright*{ \bm{y} } \< \\[-1.5ex]
        \< \sendmessageleft*{ \texttt{Open}^\collcom \prn{ \bm{e} } } \< \\[-1.5ex]
        \< \sendmessageright*{ \texttt{Open}^\collcom \prn{ \prn{ b_i^{e_i}, b_{i + 1}^{e_i} }_{i \in [n, 2]} } } \\
        \< \< \forall i \in [n, 2] \ y_i = b_i^{e_i} \oplus b_{i + 1}^{e_i} ?
    }
    \caption{The \texttt{BlobEqualities} protocol, which is a parallelized version of \texttt{BlobEquality}.}
    \label{fig:blobequalities}
\end{figure*}
Notice that, unlike in \texttt{BlobEquality}, \R{}'s commitment to $\bm{e}$ is made earlier in the \texttt{UCCompiler} protocol,
even before \S{}'s commitment.
Interestingly, this early commitment plays an important role in ensuring the UC proof holds.
\section{UC-secure commitments in the communicating malicious PUFs model}
\label{sec:ourprotocol}

\subsection{Protocol}

In Fig. \ref{fig:collective extpuf one}, we define a collective version of \ExtPUF{}, which we call
\CollExtPUF{}.
In Appendix \ref{appendix:ourproof}, we prove that this is an ideal extractable collective commitment
scheme in the \FComMPUF{}-hybrid model, if we assume that the malicious PUFs created by the
adversary
\begin{itemize}
    \item are \textbf{stateless} and have \textbf{no outgoing communication};
    \item have \textbf{unbounded incoming communication}.
\end{itemize}

Concretely, we prove the following theorem:
\begin{theorem}
    \label{collective extpuf one main text}
    \CollExtPUF{} is an ideal extractable collective commitment in the \FComMPUF{}-hybrid model, with $\kstate = \kout = 0$ and unbounded \kin{}.
\end{theorem}

Our UC-secure commitment protocol follows from applying the compiler in Fig. 
\ref{fig:revised uccompiler} to \CollExtPUF{},  as stated in the theorem below: 
\begin{theorem}
\label{thm:uc}
    Let \CollExtPUF{} be the ideal extractable collective commitment scheme in the \FComMPUF-hybrid model shown in Fig. \ref{fig:collective extpuf one}.
    Then, the protocol given by \texttt{UCCompiler} applied on  \CollExtPUF{} UC-realizes \Fcom{}
    in the \FComMPUF-hybrid model.
\end{theorem}

The proof of Theorem \ref{thm:uc} follows from Theorem \ref{thm:ucappendix} that can be found in Appendix \ref{appendix:ucproof}.
\begin{figure*}[ht]
    \procb{\CollExtPUF{}}{
        \< \S \prn{ x^1, \cdots, x^N } \< \< \R \< \\[1ex]
        \textit{Commit:   } \< \PUF \gets \puffamily \< \< \PUFE \gets \puffamily_\E \text{ with } \mathsf{TQ} \\
        \< \text{for each } i \in [N]: \< \< \\
        \< \pcind s^i \sample \bin^n \< \< \\
        \< \pcind st^i, p^i \gets \Gen \prn{ \PUF \prn{ s^i } } \< \< \\[-1.5ex]
        \< \< \sendmessageright*{ \PUF, \bm{p} } \< \\[-1.5ex]
        \< \< \sendmessageleft*{\PUFE} \< \\
        \< \text{for each } i \in [N]: \< \< \\
        \< \pcind st^i_\E, p^i_\E \gets \Gen \prn{ \PUFE \prn{ \Enc \prn{ st^i } } } \< \< \\
        \< \pcind \text{if \PUFE{} aborts, } st^i_\E, p^i_\E \coloneqq \bm{0} \< \< \\[-1.5ex]
        \< \< \sendmessageright*{\PUFE} \< \\
        \< \< \< \text{verify } \mathsf{TQ} \\
        \< \< \< \text{for each } i \in [N]: \\
        \< \< \< \pcind r^i \sample \bin^{kn} \\[-1.5ex]
        \< \< \sendmessageleft*{ \bm{r} } \< \\
        \< \text{for each } i \in [N]: \< \< \\
        \< \pcind c^i \coloneq st^i \oplus \prn{ \prn{ x^i }^n \land r^i } \< \< \\[-1.5ex]
        \< \< \sendmessageright*{ \bm{c} } \< \\
        \textit{Decommit $i$:   } \< \< \sendmessageright*{ i, s^i, x^i, st^i_\E, p^i_\E } \< \\
        \< \< \< st^i \gets \Rep \prn{ \PUF \prn{ s^i }, p^i } \\
        \< \< \< c^i = st^i \oplus \prn{ \prn{ x^i }^n \land r^i } ? \\
        \< \< \< st^i_\E = \\ 
        \< \< \< \hspace{5pt}\Rep \prn{ \PUFE \prn{ \Enc \prn{ st^i } }, p^i_\E } ? \\  
        \< \< \< \text{output } x^i
    }
    \caption{The \CollExtPUF{} protocol.}
    \label{fig:collective extpuf one}
\end{figure*}
\subsection{Efficiency improvements}

Let us examine the efficiency of the protocol \texttt{UCCompiler}.
In this protocol, \S{} performs $4n$ commitments and $3n$ decommitments, while \R{} makes one commitment and one decommitment.
We will evaluate the impact of using collective commitments in optimizing the protocol, specifically in terms of the number of PUFs
used and PUF exchange phases.

First, suppose we were to use multiple executions of the \ExtPUF{} protocol, as proposed in \cite{UCComm}.
In this setup, each commitment requires the creation of two PUFs.
Additionally, each commitment phase involves two PUF exchange phases, while no exchanges are required in the decommitment phase.
Therefore, this approach would require a total of $8n + 2$ PUFs and $8n + 2$ PUF exchange phases.

Now, let us look at the efficiency when using our \CollExtPUF{} protocol.
In this case, we employ two collective commitments -- one for each direction.
Each collective commitment requires two PUFs.
In the commitment phase, two PUF exchange phases are needed, while no exchanges are required in the decommitment phases.
Thus, this results in a total of four PUFs and four PUF exchange phases, offering a substantial improvement over the previous approach. Nevertheless, in future work one could construct a more efficient protocol that further reduces these requirements. For instance, it might be possible to achieve the same functionality using only two PUFs and two exchange phases by performing commitments in both directions simultaneously. Naturally, such a commitment scheme would need to be formally defined.

\section{Acknowledgements}
This work is funded by Fundação para a Ciência e a Tecnologia (FCT), 
Portugal FCT/MECI through national funds and when applicable
co-funded EU funds under the Unit UIDB/50008. The authors also acknowledge support from the project  PUFSeQure (\href{https://doi.org/10.54499/2023.14154.PEX}{2023.14154.PEX}) funded by FCT through national funds.

\appendix
\section{Appendices}

\subsection{Malicious PUF functionality}
\label{appendix:maliciousfunctionality}

The \FMPUF{} functionality for malicious PUFs with no communication is depicted in Fig. 
\ref{fig:mpuf functionality}.

\begin{figure}[ht]
    \centering
    \noindent\fbox{%
        \parbox{0.95\linewidth}{
            \begin{center}
                \textbf{Malicious PUF Functionality \FMPUF $(\puffamily, \kstate)$}
            \end{center}
        
            Run with parties $\mathbb{P} = \set{P_1, \cdots, P_k}$ and adversary \Sim{}.
            Create empty lists $\mathcal{L}$ and $\mathcal{M}$.
            \begin{itemize}
                \item Upon receiving $(\msf{sid}, \msf{init}, \msf{honest}, P)$ or $(\msf{sid}, \msf{init}, \msf{malicious}, M, P)$ from
                $P \in \mathbb{P} \cup \set{\Sim}$, check whether $\mathcal{L}$ contains some $(\msf{sid}, *, *, *, *)$:
                \begin{itemize}
                    \item If so, turn to the waiting state;
                    \item Else, draw $\msf{id} \gets \pufsample_n$, add $(\msf{sid}, \msf{honest}, \msf{id}, P, \bot)$
                    to $\mathcal{L}$ and send $(\msf{sid}, \msf{initialized})$ to $P$.
                    Furthermore, in the second case, add $(\msf{sid}, P, M)$ to $\mathcal{M}$.
                \end{itemize}

                \item Upon receiving $(\msf{sid}, \msf{eval}, P, s)$ from $P \in \mathbb{P} \cup \set{\Sim}$, check whether $\mathcal{L}$ contains
                $(\msf{sid}, \msf{mode}, \msf{id}, P, \bot)$ or $(\msf{sid}, \msf{mode}, \msf{id}, \bot, *)$ in case $P = \Sim{}$:
                \begin{itemize}
                    \item If it is not the case, turn to the waiting state;
                    \item Else, if $\msf{mode} = \msf{honest}$, run $\sigma \gets \pufeval_n \prn{\msf{id}, s}$
                    and send $(\msf{sid}, \msf{response}, s, \sigma)$ to $P$;
                    \item Else, if $\msf{mode} = \msf{malicious}$, get $(\msf{sid}, P, M)$ from $\mathcal{M}$, run
                    $\sigma \gets M(s)$ and send $(\msf{sid}, \msf{response}, s, \sigma)$ to $P$.
                \end{itemize}

                \item Upon receiving $(\msf{sid}, \msf{handover}, P_i, P_j)$ from $P_i$, check whether $\mathcal{L}$
                contains some $(\msf{sid}, *, *, P_i, \bot)$:
                \begin{itemize}
                    \item If it is not the case, turn to the waiting state;
                    \item Else, replace the tuple $(\msf{sid}, \msf{mode}, \msf{id}, P_i, \bot)$ in $\mathcal{L}$ with
                    $(\msf{sid}, \msf{mode}, \msf{id}, \bot, P_j)$ and send $(\msf{sid}, \msf{invoke}, P_i, P_j)$ to \Sim{}.
                \end{itemize}

                \item Upon receiving $(\msf{sid}, \msf{ready}, \Sim{})$ from \Sim{}, check whether $\mathcal{L}$ contains
                $(\msf{sid}, \msf{mode}, \msf{id}, \bot, P_j)$:
                \begin{itemize}
                    \item If it is not the case, turn to the waiting state;
                    \item Else, replace the tuple $(\msf{sid}, \msf{mode}, \msf{id}, \bot, P_j)$ in $\mathcal{L}$ with
                    $(\msf{sid}, \msf{mode}, \msf{id}, P_j, \bot)$, send $(\msf{sid}, \msf{handover}, P_i)$ to $P_j$ and add
                    $(\msf{sid}, \msf{received}, P_i)$ to $\mathcal{L}$.
                \end{itemize}

                \item Upon receiving $(\msf{sid}, \msf{received}, P_i)$ from \Sim{}, check whether $\mathcal{L}$ contains that tuple.
                \begin{itemize}
                    \item If so, send $(\msf{sid}, \msf{received})$ to $P_i$;
                    \item Otherwise, turn to the waiting state.
                \end{itemize}
            \end{itemize}
        }
    }%
    \caption{The malicious PUF functionality.}
    \label{fig:mpuf functionality}
\end{figure}
\subsection{Entropy properties}
\label{appendix:entropy}

Consider a random variable \( X \) defined on a set \( D_X \).
The \textbf{max-entropy} and \textbf{min-entropy} of \( X \) are respectively defined as
\begin{align*}
    \maxentropy{X} &= \log \prn{ \abs{ D_X } }; \\
    \minentropy{X} &= -\log \prn{ \max_{x \in D_X} \prob{X = x} }.
\end{align*}

Now, consider another random variable \( Y \), which may be correlated with \( X \).
The \textbf{average min-entropy} of \( X \) given \( Y \) is defined as
\begin{align*}
    &\condavgminentropy{X}{Y} \\
    &= - \log \prn{ \expsub{ y \gets Y }{ \max_{ x \in D_X } \condprob{X = x}{Y = y} } }.
\end{align*}

\begin{lemma}
    \label{lemma function}
    For any function $f$ and random variables $X, Y$, we have $\condavgminentropy{X}{Y} \leq \condavgminentropy{X}{f(Y)}$.
\end{lemma}
\begin{proof}
    We want to show that
    \begin{align*}
        &- \log{ \expsub{ y \gets Y }{ \max_{ x \in D_X } \condprob{X = x}{Y = y} }} \\
        \leq &- \log{ \expsub{ z \gets f(Y) }{ \max_{ x \in D_X } \condprob{X = x}{f(Y) = z} }},
    \end{align*}
    which is equivalent to
    \begin{align*}
        &\expsub{ y \gets Y }{ \max_{ x \in D_X } \condprob{X = x}{Y = y} } \\
        &\geq \expsub{ z \gets f(Y) }{ \max_{ x \in D_X } \condprob{X = x}{f(Y) = z} }.
    \end{align*}
    
    For each $z \in D_{f(Y)}$, let $x_z^*$ be such that
    \[
        \condprob{X = x_z^*}{f(Y) = z} = \max_{ x \in D_X } \condprob{X = x}{f(Y) = z}.
    \]

    Thus, 
    \begin{align*}
        &\expsub{ y \gets Y }{ \max_{ x \in D_X } \condprob{X = x}{Y = y} } \\
        \geq& \expsub{ y \gets Y }{ \condprob{X = x_{f(y)}^*}{Y = y} } \\
        =& \sum_{y \in D_Y} \prob{Y = y} \condprob{X = x_{f(y)}^*}{Y = y} \\
        =& \sum_{y \in D_Y} \prob{X = x_{f(y)}^*, Y = y} \\
        =& \sum_{z \in D_{f(Y)}} \sum_{y \in f^{-1}(z)} \prob{X = x_z^*, Y = y} \\
        =& \sum_{z \in D_{f(Y)}} \prob{X = x_z^*, f(Y) = z} \\
        =& \sum_{z \in D_{f(Y)}} \prob{f(Y) = z} \condprob{X = x_z^*}{f(Y) = z} \\
        =& \expsub{ z \gets f(Y) }{ \max_{ x \in D_X } \condprob{X = x}{f(Y) = z} }.
    \end{align*}
\end{proof}


\begin{lemma}
    \label{lemma independent}
    The following properties of the average min-entropy regarding independence hold for any
    random variables $X, Y, Z$:
    \begin{itemize}
        \item If $(X, Y) \indep Z$, then $\condavgminentropy{X}{Y, Z} = \condavgminentropy{X}{Y}$;
        \item If $X \indep Z$, then $\condavgminentropy{X}{Z} = \minentropy{X}$.
    \end{itemize}
\end{lemma}
\begin{proof}
    To prove the first property, notice that $(X, Y) \indep Z$ implies that for all $x, y, z$,
    \[
        \condprob{X = x}{Y = y, Z = z} = \condprob{X = x}{Y = y}.
    \]
        
    Thus,
    \begin{align*}
        &\condavgminentropy{X}{Y, Z} \\
        =& - \log \prn{ \expsub{ (y, z) \gets (Y, Z) }{ \max_{ x \in D_X } \condprob{X = x}{Y = y, Z = z} } } \\
        =& - \log \prn{ \expsub{ (y, z) \gets (Y, Z) }{ \max_{ x \in D_X } \condprob{X = x}{Y = y} } } \\
        =& - \log \prn{ \expsub{ y \gets Y }{ \max_{ x \in D_X } \condprob{X = x}{Y = y} } } \\
        =& \condavgminentropy{X}{Y}.
    \end{align*}

    The second property is a particular case of the first one, where $Y$ is constant.
\end{proof}


\begin{lemma}
    \label{lemma minentropy}
    Let $A$ be a random variable and $X_A$ be a random variable that is parametrized on $A$.
    Furthermore, let $Y$ be another random variable and suppose $\condavgminentropy{X_a}{Y} = H$ for all $a \in D_A$.
    Then,
    \[
        \condavgminentropy{X_A}{Y} \geq H - \maxentropy{A}.
    \]
\end{lemma}
\begin{proof}
    We just have to notice that
    \begin{align*}
        &\minentropy{X_A} \\
        =& -\log \prn{ \expsub{ y \gets Y }{ \max_{x \in D_{X_A}} \condprob{ X_A = x }{Y = y} } } \\
        =& -\log \prn{ \expsub{ y \gets Y }{ \max_{x \in D_{X_A}} \sum_{a \in D_A} \condprob{ X_a = x, A = a }{Y = y} } } \\
        \geq& -\log \prn{ \expsub{ y \gets Y }{ \max_{x \in D_{X_A}} \sum_{a \in D_A} \condprob{ X_a = x }{Y = y} } } \\
        \geq& -\log \prn{ \expsub{ y \gets Y }{ \sum_{a \in D_A} \max_{x \in D_{X_A}} \condprob{ X_a = x }{Y = y} } } \\
        =& -\log \prn{ \sum_{a \in D_A} \expsub{ y \gets Y }{ \max_{x \in D_{X_A}} \condprob{ X_a = x }{Y = y} } } \\
        =& -\log \prn{ \sum_{a \in D_A} 2^{-\condavgminentropy{X_a}{Y}} } \\
        =& -\log \prn{ \abs{D_A} 2^{-H} } \\
        =& -\log \prn{ 2^{-H} } - \log \prn{ \abs{D_A} } \\
        =& H - \maxentropy{A}. \\
    \end{align*}
\end{proof}


\begin{lemma}
    \label{lemma chain rule}
    The following weak chain rule holds for any random variables $X, Y$:
    \[
        \condavgminentropy{X}{Y, Z} \geq \condavgminentropy{X}{Y} - \maxentropy{Z}.
    \]

    Furthermore, as a consequence,
    \[
        \condavgminentropy{X}{Z} \geq \minentropy{X} - \maxentropy{Z}.
    \]
\end{lemma}
\begin{proof}
    It follows from Lemma 2.2 of \cite{FuzzyExtractors}.
\end{proof}


\begin{lemma}
    \label{lemma equality}
    Let $X$ and $Y$ be random variables defined on the same set $D$.
    Then,
    \[
        \prob{X = Y} \leq 2^{- \condavgminentropy{X}{Y}}.
    \]
\end{lemma}
\begin{proof}
    As noted in \cite{stackexchange},
    \begin{equation*}
        \begin{split}
            \prob{X = Y}
            &= \sum_{y \in D} \prob{X = y, Y = y} \\
            &= \sum_{y \in D} \condprob{X = y}{Y = y} \prob{Y = y} \\
            &= \expsub{y \gets Y}{\condprob{X = y}{Y = y}} \\
            &\leq \expsub{y \gets Y}{\max_{x \in D} \condprob{X = x}{Y = y}} \\
            &= 2^{- \condavgminentropy{X}{Y}}.
        \end{split}
    \end{equation*}
\end{proof}


\begin{lemma}
    \label{lemma neighborhood}
    Let $X$ and $Y$ be random variables defined on the same set $D$.
    Furthermore, consider a family of sets $A(d) \subseteq D$ for each $d \in D$ such that:
    \begin{itemize}
        \item $\abs{A(d)}$ does not depend on $d$ (and so we write $\abs{A}$ instead);
        \item $\bigcup_{d \in D} A(d) = D$.
    \end{itemize}
    
    Then,
    \[
        \prob{X \in A(Y)} \leq |A| \ 2^{- \condavgminentropy{X}{Y}}.
    \]
\end{lemma}
\begin{remark}
    Notice that Lemma \ref{lemma equality} is a particular case of this one, where $A(d) = \set{d}$, for $d \in D$.
\end{remark}
\begin{proof}    
    \begin{equation*}
        \begin{split}
            &\prob{X \in A(Y)} \\
            =& \sum_{y \in D} \condprob{X \in A(y)}{Y = y} \prob{Y = y} \\
            =& \expsub{y \gets Y}{\condprob{X \in A(y)}{Y = y}} \\
            \leq& \expsub{y \gets Y}{\max_{x \in D} \condprob{X \in A(x)}{Y = y}} \\
            =& \expsub{y \gets Y}{\max_{x \in D} \sum_{x' \in A(x)} \condprob{X = x'}{Y = y}} \\
            \leq& \expsub{y \gets Y}{\max_{x \in D} \sum_{x' \in A(x)} \max_{x' \in A(x)} \condprob{X = x'}{Y = y}} \\
            =& \expsub{y \gets Y}{\max_{x \in D} |A| \ \max_{x' \in A(x)} \condprob{X = x'}{Y = y}} \\
            =& |A| \ \expsub{y \gets Y}{\max_{x \in D} \max_{x' \in A(x)} \condprob{X = x'}{Y = y}} \\
            =& |A| \ \expsub{y \gets Y}{\max_{x' \in \bigcup_{x \in D} A(x)} \condprob{X = x'}{Y = y}} \\
            =& |A| \ \expsub{y \gets Y}{\max_{x' \in D} \condprob{X = x'}{Y = y}} \\
            =& |A| \ 2^{- \condavgminentropy{X}{Y}}.
        \end{split}
    \end{equation*}
\end{proof}

\subsection{CQ property}
\label{appendix:wellspread}

We formalize the CQ property as an interaction between an honest challenger \C{} and 
an adversary \Adv{}, depicted in Fig. \ref{fig:wellspreaddomain}, where \Adv{} is successful if it 
can query \PUF{} on a challenge that is close to $s$.
We say that \puffamily{} satisfies the CQ property if any adversary succeeds with 
negligible probability.
\begin{figure}[ht]
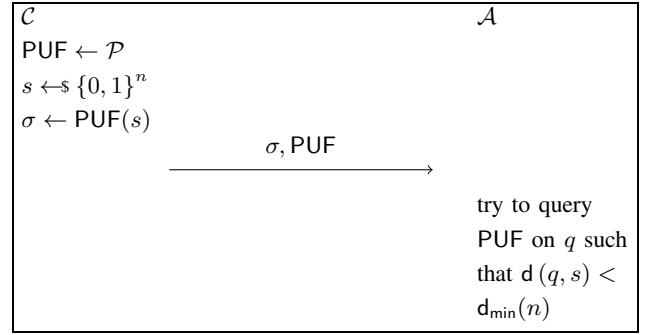

    \procb{}{
        \C \< \< \Adv \\
        \PUF \gets \puffamily \< \< \\
        s \sample \bin^n \< \< \\
        \sigma \gets \PUF(s) \< \< \\ [-1.5ex]
        \< \sendmessageright*{\sigma, \PUF} \< \\
        \< \< \text{try to query} \\
        \< \< \text{\PUF{} on $q$ such} \\
        \< \< \text{that } \dham \prn{q, s} < \\
        \< \< \dmin(n)
    }
    \caption{CQ property interaction.}
    \label{fig:wellspreaddomain}
\end{figure}

\begin{proposition}
    \label{wellspreaddomain}
    If \puffamily{} satisfies the preimage entropy property, then it satisfies the CQ property.
\end{proposition}
\begin{proof}
    Let $\bm{Q}$ be the sequence of queries \Adv{} makes to \PUF{}, which is a random variable.
    Notice that
    \begin{equation*}
        \begin{split}
            \prob{\Adv \text{ is successful}}
            &= \prob{ \dham \prn{\bm{Q}, S} < \dmin(n) } \\
            &= \prob{ \exists Q \in \bm{Q} : \dham \prn{Q, S} < \dmin(n) } \\
            &\leq \sum_{Q \in \bm{Q}} \prob{ \dham \prn{Q, S} < \dmin(n) }.
        \end{split}        
    \end{equation*}
    
    Since \Adv{} is PPT, it only queries \PUF{} a polynomial number of times and therefore it suffices to show that
    $\prob{ \dham \prn{Q, S} < \dmin(n) }$ is negligible for each $Q \in \bm{Q}$.

    First, notice that $Q = f(\PUF(S), X)$, for some function $f$ and some random variable $X \indep S$.
    Indeed, if $Q$ is the first query, this is clear, since at that point the only information that \Adv{} has that depends on $S$ is $\PUF(S)$.
    Furthermore, suppose some queries $Q'$ were done before $Q$ and all of them are of that form.
    Now, $Q$ can additionally depend on $\PUF(Q')$, but since all the $Q'$ are of that form, the same will happen with $Q$.
    From this observation, we get
    \begin{align*}
        \prob{ \dham(Q, S) < \dmin(n) }
        &= \prob{ S \in B_n^{\dmin(n)}(Q) } \\
        &\leq \left|B_n^{\dmin(n)} \right| 2^{-\condavgminentropy{S}{Q}} \tag{Lemma \ref{lemma neighborhood}} \\
        &\leq \left|B_n^{\dmin(n)} \right| 2^{-\condavgminentropy{S}{\PUF(S), X}} \tag{Lemma \ref{lemma function}} \\
        &\leq \left|B_n^{\dmin(n)} \right| 2^{-\condavgminentropy{S}{\PUF(S)}}, \tag{Lemma \ref{lemma independent}}
    \end{align*}
    which is negligible, by assumption.
\end{proof}

\begin{remark}
    In \cite{PUFsUC}, it was assumed that $\dmin(n) \in o(n / \log(n))$, which corresponds to a simplified version of the preimage entropy property.
    Indeed, first notice that
    \begin{align*}
        \abs{B_n^{\dmin(n)}} &= \sum_{k = 0}^{\dmin(n) - 1} \binom{n}{k} \\
        &\leq \sum_{k = 0}^{\dmin(n) - 1} n^k \\
        &= \frac{n^{\dmin(n)} - 1}{n - 1} \\
        &\leq n^{\dmin(n)}.
    \end{align*}

    Since $\condavgminentropy{S}{\PUF(S)}$ was not taken into consideration
    (that is, it was assumed to take its maximum value $n$), this then leads to
    \begin{equation*}
        \begin{split}
            \abs{B_n^{\dmin(n)}} 2^{-\condavgminentropy{S}{\PUF(S)}}
            &\leq n^{\dmin(n)} 2^{-n} \\
            &\leq 2^{-n + \dmin(n) \log(n)} \\
            &= 2^{-n + o(n)},
        \end{split}
    \end{equation*}
    which is negligible.
\end{remark}
\subsection{Indistinguishability property}
\label{appendix:indist}

We can formalize the indistinguishability property as an interaction between an honest challenger \C{} and a distinguisher \D{}, depicted in Fig. \ref{fig:single indist},
where \D{} is successful if it can guess whether the extracted string it received comes from \PUF{} or is uniform.
We say that \puffamily{} satisfies the indistinguishability property if all distinguishers succeed with probability negligibly close to $\frac{1}{2}$,
that is,
\[
    \prob{\D = B} = \frac{1}{2} + \varepsilon(n).
\]
\begin{figure}[ht]
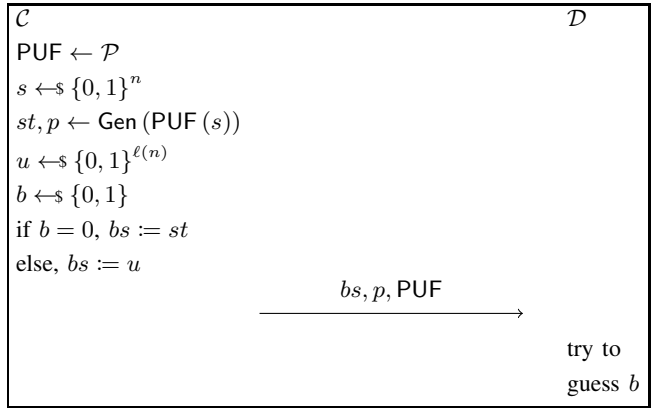

    \procb{}{
        \C \< \< \D \\
        \PUF \gets \puffamily \< \< \\
        s \sample \bin^n \< \< \\
        st, p \gets \Gen \prn{ \PUF \prn{ s } } \< \< \\
        u \sample \bin^{ \fel(n) } \< \< \\
        b \sample \bin \< \< \\
        \text{if $b = 0$, } \bs \coloneqq st \< \< \\
        \text{else, } \bs \coloneqq u \< \< \\[-1.5ex]
        \< \sendmessageright*{\bs, p, \PUF} \< \\
        \< \< \text{try to}\\
        \< \< \text{guess } b
    }
    \caption{Indistinguishability property interaction.}
    \label{fig:single indist}
\end{figure}

\begin{lemma}
    \label{single auxprop}
    Consider the indistinguishability property interaction depicted in Fig. \ref{fig:single indist} and
    let $\bm{Q}$ denote the sequence of queries \D{} makes to \PUF{}.
    Then, for all distinguishers \D{},
    \[
        \condprob{\D = B}{\dham \prn{\bm{Q}, S} \geq \dmin(n)} = \frac{1}{2} + \varepsilon(n).
    \]
\end{lemma}
\begin{proof}
    Suppose $\bm{Q} = \bm{q}$ and $S = s$ such that $\dham(\bm{q}, s) \geq \dmin(n)$.
    Consider the programs depicted in Fig. \ref{fig:procedures single auxprop}.
    \begin{figure}[ht]
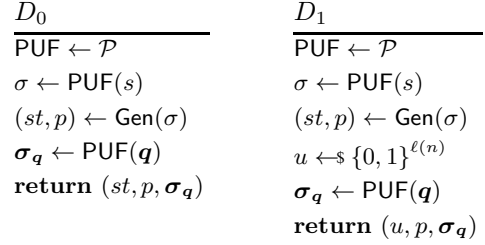

        \begin{pchstack}[center, space = 1cm]
            \procedure{$D_0$}{
                \PUF \gets \puffamily \\
                \sigma \gets \PUF(s) \\
                (st, p) \gets \Gen(\sigma) \\
                \bm{\sigma}_{\bm{q}} \gets \PUF(\bm{q}) \\
                \pcreturn (st, p, \bm{\sigma}_{\bm{q}})
            }
    
            \procedure{$D_1$}{
                \PUF \gets \puffamily \\
                \sigma \gets \PUF(s) \\
                (st, p) \gets \Gen(\sigma) \\
                u \sample \bin^{\fel(n)} \\
                \bm{\sigma}_{\bm{q}} \gets \PUF(\bm{q}) \\
                \pcreturn (u, p, \bm{\sigma}_{\bm{q}})
            }
        \end{pchstack}
        \caption{Programs $D_0$ and $D_1$.}
        \label{fig:procedures single auxprop}
    \end{figure}    

    We get $\condavgminentropy{\PUF(s)}{\PUF(\bm{q})} \geq \entropybound(n)$ from unpredictability.
    Thus, from almost-uniformity with $\EXTRA = \PUF(\bm{q})$, we know that 
    \[
        \SD \prn{ D_0, D_1 } = \abs{ \varepsilon(n) }.
    \]  
    
    Furthermore, notice that each $D_b$ corresponds to the information \D{} gets in the interaction if $B = b$.
    Therefore, since statistical closeness implies indistinguishability, we get
    \[
        \condprob{\D = B}{\bm{Q} = \bm{q}, S = s} = \frac{1}{2} + \varepsilon(n).
    \]

    Thus,
    \begin{align*}
        &\prob{\D = B, \dham \prn{\bm{Q}, S} \geq \dmin(n)} \\
        =& \sum_{\substack{(\bm{q}, s) : \\ \dham(\bm{q}, s) \geq \dmin(n)}} \condprob{\D = B}{\bm{Q} = \bm{q}, S = s} \prob{\bm{Q} = \bm{q}, S = s} \\
        =& \sum_{\substack{(\bm{q}, s) : \\ \dham(\bm{q}, s) \geq \dmin(n)}} \prn{ \frac{1}{2} + \varepsilon(n) } \prob{\bm{Q} = \bm{q}, S = s} \\
        =& \prn{ \frac{1}{2} + \varepsilon(n) } \prob{\dham \prn{\bm{Q}, S} \geq \dmin(n)},
    \end{align*}
    that is,
    \[
        \condprob{\D = B}{\dham \prn{\bm{Q}, S} \geq \dmin(n)} = \frac{1}{2} + \varepsilon(n).
    \]
\end{proof}

\begin{proposition}
    \label{single indistinguishability}
    If \puffamily{} satisfies the preimage entropy property, then it satisfies the indistinguishability 
    property.
\end{proposition}
\begin{proof}
    Let $\bm{Q}$ be the sequence of queries $\D$ makes to $\PUF$, which is a random variable.
    Notice that
    \begin{align*}
        &\prob{\D = B} \\
        = &\prob{\D = B, \dham \prn{\bm{Q}, S} < \dmin(n)} + \\
        &\prob{\D = B, \dham \prn{\bm{Q}, S} \geq \dmin(n)}.
    \end{align*}

    From Lemma \ref{wellspreaddomain}, we have
    \[
        \prob{\dham \prn{\bm{Q}, S} < \dmin(n)} = \varepsilon(n).
    \]

    Furthermore, from Lemma \ref{single auxprop}, we have
    \begin{align*}
        &\prob{\D = B, \dham \prn{\bm{Q}, S} \geq \dmin(n)} \\
        = &\prn{ \frac{1}{2} + \varepsilon(n) } \prob{\dham \prn{\bm{Q}, S} \geq \dmin(n)}.
    \end{align*}

    Thus, we can conclude that
    \begin{equation*}
        \begin{split}
            \prob{\D = B}
            &= \varepsilon(n) + \prn{\frac{1}{2} + \varepsilon(n)} (1 - \varepsilon(n)) \\
            &= \frac{1}{2} + \varepsilon(n).
        \end{split}
    \end{equation*}
\end{proof}

We now present some analogous properties that are necessary for collective commitment schemes.
Let $\IND_n$ denote an arbitrary program that generates indices, where $\abs{\IND_n}$ is polynomial in $n$.
The interaction depicted in Fig. \ref{fig:collective one indist} is a generalization
of the indistinguishability interaction depicted in Fig. \ref{fig:single indist}.
Just like before, we say that \puffamily{} satisfies the collective indistinguishability property if 
all distinguishers succeed with probability negligibly close to $\frac{1}{2}$.
\begin{figure*}[ht]
    \procb{}{
            \C \< \< \D \\
            \PUF \gets \puffamily \< \< \\
            \text{for each } i \in \IND_n: \< \< \\
            \pcind s^i \sample \bin^n \< \< \\
            \pcind st^i, p^i \gets \Gen \prn{ \PUF \prn{ s^i } } \< \< \\
            \pcind u^i \sample \bin^{ \fel(n) } \< \< \\
            b \sample \bin \< \< \\
            \text{if $b = 0$, } \bm{\bs} \coloneqq \bm{st} \< \< \\
            \text{else, } \bm{\bs} \coloneqq \bm{u} \< \< \\[-1.5ex]
            \< \sendmessageright*{\bm{bs}, \bm{p}, \PUF} \< \\
            \< \< \text{try to guess } b
        }
        \caption{Collective indistinguishability interaction.}
        \label{fig:collective one indist}
\end{figure*}

\begin{lemma}
    \label{collective one auxprop}
    Consider the interaction depicted in Fig. \ref{fig:collective one indist}.    
    Let $\bm{Q}$ be the sequence of queries \D{} makes to \PUF{} and let $\bm{Q}^i \coloneqq \bm{Q} \mathbin \Vert \bm{S}^{j \neq i}$
    for each $i \in \IND_n$.
    Furthermore, we write $\varphi(\bm{q}, \bm{s})$ instead of
    \[
        \forall i \in \IND_n \ \dham \prn{\bm{q}^i, s^i} \geq \dmin(n).
    \]
    
    Then, for all distinguishers \D{},
    \[
        \condprob{\D = B}{\varphi(\bm{Q}, \bm{S})} = \frac{1}{2} + \varepsilon(n).
    \]
\end{lemma}
\begin{proof}
    Suppose $\bm{Q} = \bm{q}$ and $S = s$ that satisfy $\varphi(\bm{q}, \bm{s})$.
    For simplicity, assume $\IND_n = \set{1, \cdots, p(n)}$, for some polynomial $p(n)$.
    For each $i \in \set{0, \cdots, p(n)}$, consider the procedure depicted in Fig. \ref{fig:procedures collective one auxprop}.
    \begin{figure}[ht]
        \centering
        \procedure{$D_i$}{
            \PUF \gets \puffamily \\
            \pcfor 1 \leq j \leq p(n) : \\
            \pcind \sigma^j \gets \PUF \prn{ s^j } \\
            \pcind \prn{st^j, p^j} \gets \Gen \prn{ \sigma^j } \\
            \pcfor 1 \leq j \leq i: \\
            \pcind u^j \sample \bin^{\fel(n)} \\
            \bm{\sigma}_{\bm{q}} \gets \PUF \prn{ \bm{q} } \\
            \pcreturn \prn{ \bm{u}^{j \leq i}, \bm{st}^{j > i}, \bm{p}, \bm{\sigma}_{\bm{q}} }
        }
        \caption{Procedures $D_i$ for $i \in \set{0, \cdots, p(n)}$.}
        \label{fig:procedures collective one auxprop}
    \end{figure}

    Let us focus on the difference between $D_{i-1}$ and $D_i$.    
    In the first one, $ST^i$ is used, while in the second one, $U^i$ is used instead.
    Furthermore, both have the additional information
    \[
        \EXTRA_i \coloneqq \prn{ \bm{U}^{j < i}, \bm{ST}^{j > i}, \bm{P}^{j \neq i}, \PUF \prn{ \bm{q} } }.
    \]    
    
    Notice that $\EXTRA_i$ is a function of $\PUF \prn{\bm{q}^i}$ and some $X \indep \PUF \prn{s^i}$.
    From that fact and from unpredictability, we have 
    \begin{align*}
        &\condavgminentropy{\PUF \prn{s^i}}{\EXTRA_i} \\
        =& \condavgminentropy{\PUF \prn{s^i}}{\PUF \prn{\bm{q}^i}} \\
        \geq& \entropybound(n).
    \end{align*}
    
    Therefore, from almost-uniformity with $\EXTRA_i$, we know that $\SD \prn{ D_{i-1}, D_i } = \abs{ \varepsilon(n) }$ for each $i \in \IND_n$.
    Furthermore, since $p(n)$ is a polynomial,
    \begin{align*}
        &\SD \prn{ D_0, D_{p(n)} } \\
        \leq& \sum_{i = 1}^{p(n)} \SD \prn{ D_{i-1}, D_i } \\
        =& p(n) \abs{ \varepsilon(n) } \\
        =& \abs{ \varepsilon(n) }.
    \end{align*}
    
    Notice that $D_0$ corresponds to the information \D{} gets in the interaction if $B = 0$;
    the same can be said about $D_{p(n)}$ if $B = 1$.
    Therefore, just like in Lemma \ref{single auxprop}, we get
    \[
        \condprob{\D = B}{\bm{Q} = \bm{q}, \bm{S} = \bm{s}} = \frac{1}{2} + \varepsilon(n).
    \]

    Thus,
    \begin{align*}
        &\prob{\D = B, \varphi(\bm{Q}, \bm{S})} \\
        =& \sum_{\substack{(\bm{q}, \bm{s}) : \\ \varphi(\bm{q}, \bm{s})}} \condprob{\D = B}{\bm{Q} = \bm{q}, \bm{S} = \bm{s}} \prob{\bm{Q} = \bm{q}, \bm{S} = \bm{s}} \\
        =& \sum_{\substack{(\bm{q}, \bm{s}) : \\ \varphi(\bm{q}, \bm{s})}} \prn{ \frac{1}{2} + \varepsilon(n) } \prob{\bm{Q} = \bm{q}, \bm{S} = \bm{s}} \\
        =& \prn{ \frac{1}{2} + \varepsilon(n) } \prob{\varphi(\bm{Q}, \bm{S})},
    \end{align*}
    that is,
    \[
        \condprob{\D = B}{\varphi(\bm{Q}, \bm{S})} = \frac{1}{2} + \varepsilon(n).
    \]
\end{proof}

\begin{lemma}
    \label{collective one indistinguishability}
    If \puffamily{} satisfies the preimage entropy property, then it satisfies the collective 
    indistinguishability property.
\end{lemma}
\begin{proof}
    Notice that
    \[
        \prob{\D = B} = \prob{\D = B, \neg \varphi(\bm{S}, \bm{Q})} + \prob{\D = B, \varphi(\bm{S}, \bm{Q})},
    \]
    where $\varphi(\bm{S}, \bm{Q})$ is the event defined in Lemma \ref{collective one auxprop}.

    From Lemma \ref{wellspreaddomain}, we have
    \[
        \prob{\neg \varphi(\bm{S}, \bm{Q})} = \varepsilon(n).
    \]

    Furthermore, from Lemma \ref{collective one auxprop}, we have
    \[
        \prob{\D = B, \varphi(\bm{S}, \bm{Q})} = \prn{ \frac{1}{2} + \varepsilon(n) } \prob{\varphi(\bm{S}, \bm{Q})}.
    \]

    Thus, we can conclude that
    \begin{equation*}
        \begin{split}
            \prob{\D = B}
            &= \prob{\D = B, \neg \varphi(\bm{S}, \bm{Q})} + \prob{\D = B, \varphi(\bm{S}, \bm{Q})} \\
            &= \varepsilon(n) + \prn{\frac{1}{2} + \varepsilon(n)} (1 - \varepsilon(n)) \\
            &= \frac{1}{2} + \varepsilon(n).
        \end{split}
    \end{equation*}
\end{proof}
\subsection{CRP guessing property}
\label{appendix:crp}

We can formalize the CRP guessing property with an interaction between an honest challenger \C{} and an adversary \Adv{}, as depicted in Fig. \ref{fig:guesscrp},
where \Adv{} is successful if $st = \Rep(\PUF(s), p)$ and throughout its execution it does not query \PUF{} on any challenge $q$
such that $\dham(q, s) < \dmin(n)$.
We say that \puffamily{} satisfies the CRP guessing property if all adversaries \Adv{} in this interaction succeed with negligible probability.
\begin{figure}[ht]
    \procb{}{
        \C \< \< \Adv \\
        \PUF \gets \puffamily \< \< \\ [-1.5ex]
        \< \sendmessageright*{\PUF} \< \\
        \< \< \text{output $s, st, p$}
    }
    \caption{CRP guessing property interaction.}
    \label{fig:guesscrp}
\end{figure}

Although it seems like it can be reduced to the indistinguishability property, we were not able to prove this result.
It seems that such a reduction would depend on specific details of the FE, which we want to avoid.

Alternatively, one might consider reducing this requirement to a more fundamental property involving only the PUF family, as depicted in Fig.
\ref{fig:alternative guesscrp}, where \Adv{} is successful if $\sigma$ is a possible response for $\PUF(s)$.
However, achieving this reduction is challenging, as it would also rely on the specific characteristics of the FE
and PUF family in use.
Indeed, even if we could identify some $w$ such that $\Gen(w) = (st, p)$, it would not necessarily enable us to recover an actual possible response
$\sigma$ for $\PUF(s)$.

\begin{figure}[ht]
    \procb{}{
        \C \< \< \Adv \\
        \PUF \gets \puffamily \< \< \\ [-1.5ex]
        \< \sendmessageright*{\PUF} \< \\
        \< \< \text{output $s, \sigma$}
    }
    \caption{Alternative CRP guessing property interaction.}
    \label{fig:alternative guesscrp}
\end{figure}

\subsection{Test query property}
\label{appendix:testquery}
The test query property can be formalized with an interaction between an honest challenger \C{} and an adversary \Adv{}, as depicted in Fig. \ref{fig:testquery},
where \Adv{} is successful if it sends $\PUF^* \neq \PUF$ such that $\Rep(\PUF^*(s), p) = st$.
We say that \puffamily{} satisfies the test query property if all adversaries \Adv{} in this interaction succeed with negligible probability.
This is proved in Proposition \ref{testquery}.
\begin{figure}[ht]
    \procb{}{
        \C \< \< \Adv \\
        \PUF \gets \puffamily \< \< \\
        s \sample \bin^n \< \< \\
        st, p \gets \Gen(\PUF(s)) \< \< \\[-1.5ex]
        \< \sendmessageright*{\PUF} \< \\[-1.5ex]
        \< \sendmessageleft*{\PUF^*} \< \\
        \Rep \prn{ \PUF^*(s), p } = st? \< \<
    }
    \caption{Test query property interaction.}
    \label{fig:testquery}
\end{figure}
\begin{proposition}
    \label{testquery}
    If \puffamily{} satisfies the preimage entropy property and the CRP guessing property,
    then it satisfies the test query property.
\end{proposition}
\begin{proof}
    Suppose there exists an adversary $\Adv^0$ that is successful in the test query interaction with non-negligible probability.
    In Fig. \ref{fig:reductioncrp}, we construct an adversary \Adv{} that contradicts the CRP guessing property.
    \begin{figure}[ht]
        \begin{center}
            \noindent\fbox{%
                \parbox{0.95\linewidth}{
                    \begin{center}
                        \textbf{Adversary} \Adv
                    \end{center}
                    
                    \begin{itemize}
                        \item When receiving \PUF{} from \C{}, simulate the interaction depicted in Fig. \ref{fig:guesscrp} between
                        a challenger $\C^0$ and the adversary $\Adv^0$ using \PUF{};
                        \item When $\Adv^0$ sends $\PUF^*$, run $(st, p) \gets \Gen \prn{ \PUF^*(s) }$, where $s$ is the challenge
                        that $\C^0$ created;
                        \item Output $(s, st, p)$.
                    \end{itemize}
                }
            }%
        \end{center}
        \caption{Reduction to the CRP guessing property.}
        \label{fig:reductioncrp}
    \end{figure}
    
    Let us consider what happens when $\Adv^0$ is successful and it does not query \PUF{} close to $s$.
    Notice that \Adv{}’s queries are the same as those of $\Adv^0$ to \PUF{}, along with the additional query $\PUF^*(s)$.
    Since $\PUF^* \neq \PUF$, we know that \Adv{} also does not query \PUF{} close to $s$.
    Indeed, this holds because even if $\PUF^*$ is malicious, the honest PUFs embedded within it are 
    freshly created using the honest \pufsample{} algorithm, and therefore distinct from \PUF{}.
    Moreover, since $\Adv^0$ is successful, \Adv{} outputs a valid CRP.
    Therefore, under these conditions, \A{} is always successful.

    Now, notice that $\Adv^0$ is successful with non-negligible probability.
    Furthermore, from Proposition \ref{wellspreaddomain}, we know that it only queries \PUF{} close to $s$ with negligible probability.
    Thus, we can conclude that the probability of both these conditions holding simultaneously is non-negligible, and so
    \Adv{} is successful with non-negligible probability.
\end{proof}
\subsection{Improved proofs for the previous protocols}
\label{appendix:newproofs}

The \CPUF{} protocol from \cite{MalPUFs} is depicted in Fig. \ref{fig:single nocomm cpuf}.
We consider a PUF family \puffamily{} and a fuzzy extractor family $\fe = (\Gen, \Rep)$ with matching 
parameters such that $\fel_\fe(n) = kl$, with $l(n) = 3n$.
\begin{figure*}[ht]
    \procb{\CPUF}{
        \< \S(x) \< \< \R \\[1ex]
        \textbf{Commit:   } \< \PUF \gets \puffamily{} \< \< \\
        \< s \sample \bin^n \< \< \\
        \< (st, p) \gets \Gen (\PUF(s)) \< \< \\[-1.5ex]
        \< \< \sendmessageright*{\PUF, p} \< \\
        \< \< \< r \sample \bin^{kl} \\[-1.5ex]
        \< \< \sendmessageleft*{r} \< \\
        \< c \coloneq st \oplus \prn{ x^l \land r } \< \< \\[-1.5ex]
        \< \< \sendmessageright*{c} \< \\
        \textbf{Decommit:   } \< \< \sendmessageright*{(s, x)} \< \\
        \< \< \< st \gets \Rep (\PUF(s), p) \\
        \< \< \< c = st \oplus \prn{ x^l \land r } ?\\
        \< \< \< \text{output } x
    }
    \caption{The ideal commitment scheme \CPUF{} in the \FMPUF{}-hybrid model.}
    \label{fig:single nocomm cpuf}
\end{figure*}

\begin{theorem}
    \label{single nocomm cpuf}
    \CPUF{} is an ideal commitment scheme in the \FMPUF{}-hybrid model,
    with unbounded \kstate{}.
\end{theorem}
\begin{proof}
    Completeness clearly follows from the response consistency property \ref{subsec:additionalproperties},
    and thus its proof is omitted.
    We will do the same for the remaining proofs.
    
    \textbf{Computationally hiding:}
    
    Suppose, by contradiction, that this is not the case.
    Then, there exist different strings $x^0$ and $x^1$ and a malicious receiver \Radv{} such that, in the
    interaction depicted in Fig. \ref{fig:single hiding},
    \[
        \prob{\Radv \prn{ \commit^\CPUF \prn{x^{B^0}} } = B^0} - \frac{1}{2}
    \]
    is not negligible.
    \footnote{
        We denote the random variable as $B^0$ in this interaction to avoid confusion with the analogous $B$
        in the subsequent interaction we define.
    }

    Consider a modified protocol \CUnif{},
    \footnote{
        Given that $u$ is chosen uniformly, the decommitment phase of \CUnif{} is actually not defined.
        This is not a problem, because we will only run its commitment phase.
    }
    where \S{} uses $u \sample \bin^{kl}$ instead of $st$.
    That is, in that protocol, \S{} does $c \coloneq u \oplus \prn{ x^l \land r }$.
    Then,    
    \[
        \prob{\Radv \prn{ \commit^\CUnif \prn{x^{B^0}} } = B^0} = \frac{1}{2},
    \]
    because the distribution of the interaction is independent from $B^0$.
    
    Now, consider the interaction depicted in Fig. \ref{fig:single indist} between a challenger \C{} and a distinguisher \D{},
    corresponding to the indistinguishability property of PUFs.        
    In Fig. \ref{fig:distinguisher single nocomm cpuf hiding}, we present a distinguisher \D{} that breaks this property,
    thereby contradicting Lemma \ref{single indistinguishability}.
    Here, \CBS{} denotes a protocol just like \CPUF{}, but where
    \S{} does not create its \PUF{} and does not need to query \PUF{} to get $st$ and $p$.
    Instead, \S{} receives $\bs, p, \PUF$ from \D{} (who receives them from \C{}) and uses $\bs$ instead of $st$ in $c$.
    \begin{figure}[ht]
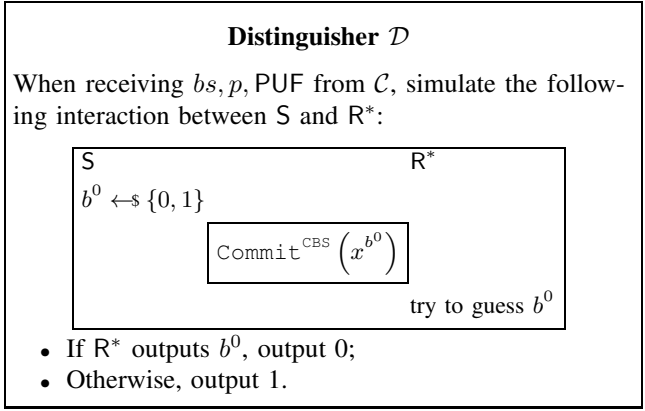

        \begin{center}
            \noindent\fbox{%
                \parbox{0.95\linewidth}{
                    \begin{center}
                        \textbf{Distinguisher} \D
                    \end{center}
            
                    When receiving $\bs, p, \PUF$ from \C{}, simulate the following interaction between \S{} and \Radv{}:
                    \procb{}{
                        \S \< \< \Radv \\
                        b^0 \sample \bin \< \< \\
                        \< \pcbox{\commit^\CBS \prn{x^{b^0}} } \< \\
                        \< \< \text{try to guess } b^0
                    }
                    \begin{itemize}
                        \item If \Radv{} outputs $b^0$, output 0;
                        \item Otherwise, output 1.
                    \end{itemize}
                }
            }%
        \end{center}
        \caption{Distinguisher that breaks the indistinguishability property.}
        \label{fig:distinguisher single nocomm cpuf hiding}
    \end{figure}
    
    Notice that if $b = 0$, then $\bs = st$, making \CBS{} identical to \CPUF{}.
    Likewise, if $b = 1$, then $\bs = u$, and so \CBS{} is identical to \CUnif{}.
    Thus,
    \begin{align*}
        &2\prob{\D = B} \\
        =& \condprob{\D = 0}{B = 0} + \condprob{\D = 1}{B = 1} \\
        =& \prob{\Radv \prn{\commit^\CPUF \prn{x^{B^0}} } = B^0} + \\
        &\prob{\Radv \prn{\commit^\CUnif \prn{x^{B^0}} } \neq B^0} \\
        =& \prob{\Radv \prn{\commit^\CPUF \prn{x^{B^0}} } = B^0} +
        1 - \frac{1}{2} \\
        =& \prob{\Radv \prn{\commit^\CPUF \prn{x^{B^0}} } = B^0} +
        \frac{1}{2},
    \end{align*}
    and so
    \begin{align*}
        &\prob{\D = B} - \frac{1}{2} \\
        =& \frac{1}{2} \prn{ \prob{\Radv \prn{\commit^\CPUF \prn{x^{B^0}} } = B^0} + \frac{1}{2} } - \frac{1}{2} \\
        =& \frac{1}{2} \prn{ \prob{\Radv \prn{\commit^\CPUF \prn{x^{B^0}} } = B^0} - \frac{1}{2} } \\
    \end{align*}
    is non-negligible, which is a contradiction.

    \textbf{Statistically Binding:}
    
    Suppose \Sadv{} has some malicious behavior and makes a certain commitment.
    We have to show that the probability of \Sadv{} being able to open that commitment to different strings $X$ and $Y$ successfully is negligible.
    
    When decommitting to $X$, \Sadv{} sends $S$, and \R{} runs $\mathsf{Rep}(\PUF(S), P)$, with this final expression being a
    random variable that depends on:
    \footnote{
        The use of capital letters is deliberate.
        Here, instead of writing the values as $x, s, p$ like in the protocol,
        we write $X, S, P$ to emphasize the fact that they are random variables.
        We adopt this convention throughout our work.
    }
    \begin{enumerate}
        \item the random variable $S$;
        \item \PUF{}'s state \STATE{} in that moment and its internal randomness, which is independent from $R$;
        \item the random variable $P$, which is independent from $R$;
        \item the randomness of the fuzzy extractor, which is independent from $R$.
    \end{enumerate}

    Therefore, we can write $\mathsf{Rep}(\PUF(S), P) = f \prn{S, \STATE, T}$, where $T$ represents some internal randomness that is
    independent from the remaining variables.        
    Thus, if the decommitment is successful,
    \[
        C = f \prn{S, \STATE, T} \oplus \prn{ X^l \land R }.
    \]
    
    Likewise, for the decommitment of $Y$, we can write $\mathsf{Rep} \prn{\PUF \prn{S'}, P} = g \prn{S', \STATE', T'}$, where $T'$
    is also independent from the remaining variables.
    If the decommitment is successful,
    \[
        C = g \prn{S', \STATE', T'} \oplus \prn{ Y^l \land R }.
    \]
    
    Suppose $X$ and $Y$ differ on an index $J$.
    Given $j \in [k]$, let ${\bm{I}}_j$ be the set of the positions of the commitment $C$ that are used for committing the $j$-th bit
    of the string.
    That is, ${\bm{I}_j} = \set{ j, j+k, \cdots, j + (l-1)k }$.
    Then,        
    \begin{align*}
        C_{\bm{I}_J} &= 
        f \prn{S, \STATE, T}_{\bm{I}_J} \oplus \prn{ X_J^l \land R_{\bm{I}_J} } \\
        &= g \prn{S', \STATE', T'}_{\bm{I}_J} \oplus \prn{ Y_J^l \land R_{\bm{I}_J} },
    \end{align*}        
    which implies        
    \[
        R_{\bm{I}_J} = f \prn{S, \STATE, T}_{\bm{I}_J} \oplus g \prn{S', \STATE', T'}_{\bm{I}_J},
    \]
    since $X_J \neq Y_J$.

    Hence, if we show that
    \[
        \prob{R_{\bm{I}_J} = f \prn{S, \STATE, T}_{\bm{I}_J} \oplus g \prn{S', \STATE', T'}_{\bm{I}_J}}
    \]         
    is negligible, we are done.

    Let $\STATE_0$ be \PUF{}'s state when it is sent to \R{}.
    Notice that $\STATE = \STATE' = \STATE_0$, given that \PUF{} is not queried up to the moment of decommitment.
    Furthermore, $\STATE_0 \indep R$, since $R$ is only sent after \PUF{}.
    Therefore,
    \begin{align*}
        &\condavgminentropy{R_{\bm{I}_J}}{f \prn{S, \STATE, T}_{\bm{I}_J} \oplus g \prn{S', \STATE', T'}_{\bm{I}_J}} \\
        =& \ \condavgminentropy{R_{\bm{I}_J}}{f \prn{S, \STATE_0, T}_{\bm{I}_J} \oplus g \prn{S', \STATE_0, T'}_{\bm{I}_J}} \\
        \geq& \ \condavgminentropy{R_{\bm{I}_J}}{\STATE_0, J, S, S', T, T'} \tag{Lemma \ref{lemma function}} \\
        =& \ \condavgminentropy{R_{\bm{I}_J}}{\STATE_0, J, S, S'} \tag{$\prn{T, T'} \indep \prn{R_{\bm{I}_J}, J, \STATE_0, S, S'}$ and Lemma \ref{lemma independent}} \\
        \geq& \ \condavgminentropy{R_{\bm{I}_J}}{\STATE_0} - \maxentropy{J, S, S', J} \tag{Lemma \ref{lemma chain rule}}.
    \end{align*}

    Now, although $\STATE_0 \indep R$, this does not necessarily imply $\STATE_0 \indep R_{\bm{I}_J}$.
    Indeed, $R_{\bm{I}_J}$ depends on $J$, which in turn can depend on $\STATE_0$.
    However, we still know that $\STATE_0 \indep R_{\bm{I}_j}$ for each $j \in [k]$.
    Therefore, by Lemma \ref{lemma independent}, for each $j \in [k]$,
    \[
        \condavgminentropy{R_{\bm{I}_j}}{\STATE_0} = \minentropy{R_{\bm{I}_j}} = l(n),
    \]
    which implies
    \begin{align*}
        &\condavgminentropy{R_{\bm{I}_J}}{\STATE_0} - \maxentropy{J, S, S', J} \\
        \geq& \ l(n) - \maxentropy{J} - \maxentropy{J, S, S', J} \tag{Lemma \ref{lemma minentropy}} \\
        =& \ 3n - 2\log(k) - 2n \\
        =& \ n - 2\log(k).
    \end{align*}

    Thus, by Lemma \ref{lemma equality},
    \begin{align*}
        &\prob{R_{\bm{I}} = f \prn{S, \STATE, T}_{\bm{I}} \oplus g \prn{S', \STATE', T'}_{\bm{I}}} \\
        \leq& 2^{-n + 2\log(k)},
    \end{align*}
    which is negligible.
\end{proof}

In what follows, we will need the definition below, adapted from \cite{UCComm}:
\begin{definition}
    \label{def:ecc}
    An \textbf{$(N, L, D)$-error-correcting code (ECC)} is a tuple of PPT algorithms $(\Enc, \Dec)$,
    where $\Enc : \bin^N \to \bin^L$ and $\Dec: \bin^L \to \bin^N$, such that:
    \begin{itemize}
        \item \textbf{Minimum distance:}
        For all messages $m_1, m_2 \in \bin^N$, the corresponding codewords are such that $\dham (\Enc(m_1), \Enc(m_2)) \geq D$.
        \footnote{
            Consequently, $D$ is called the \textit{minimum distance} of the code.
        }
        \item \textbf{Correct decoding:}
        Let $m \in \bin^N$ and $c = \Enc(m)$.
        Then, for all $c' \in \bin^L$,
        \[
            \dham \prn{ c, c' } \leq \floor{ \frac{D-1}{2} } \implies \Dec \prn{c'} = m.
        \]
    \end{itemize}
\end{definition}

\begin{theorem}
    \label{single nocomm extpuf}
    \ExtPUF{} is an ideal extractable commitment scheme in the \FMPUF{}-hybrid model, with $\kstate = 0$.
\end{theorem}
\begin{proof}
    Let $\bm{I}_j \coloneqq \set{ j, j+k, \cdots, j + (n-1)k }$, which is the set of positions of $c$ that are used for committing the
    $j$-th bit of $x$.

    \textbf{Computationally Hiding:}
    
    Notice that \S{} does not abort when \PUFE{} aborts and \PUFE{} is stateless.
    Thus, \Radv{} gets the same information in this protocol as in \CPUF{}, which is computationally hiding. 
    
    \textbf{Statistically Binding:}
    
    Suppose \Sadv{} decommits successfully with $S, X, ST_\E, P_\E$.        
    From Lemma \ref{testquery}, we know that \Sadv{} returned the same \PUFE{} with overwhelming probability, so we can assume that is the case.       

    First, we will show that with overwhelming probability \Sadv{} queried \PUFE{} on some $Q_X$ that is close to $Q = \Enc(ST)$,
    \footnote{
        That is, $\dham (Q_X, Q) < \dmin$.
    }
    where $ST$ is the one obtained by \R{} in the decommitment.
    Indeed, suppose that, with non-negligible probability, none of the queries were close to $Q$.
    In that case, we can define an adversary \Adv{}, depicted in Fig. \ref{fig:adversary extpuf guesscrp}, which contradicts 
    the CRP guessing property depicted in Fig. \ref{fig:guesscrp}, when interacting with a challenger \C{}.
    We know that
    \begin{itemize}
        \item with non-negligible probability, \Sadv{} does not query \PUFE{} with queries that are close to $Q$,
        which implies that the same happens for \Adv{};
        \item \Sadv{} decommits successfully to $X$, which implies that \Adv{} can output $Q, ST, P$ such that $ST = \msf{Rep} \prn{ \PUF(Q), P}$,
    \end{itemize}
    and so \Adv{} is successful with non-negligible probability.
    \begin{figure}
        \begin{center}
            \noindent\fbox{%
                \parbox{0.95\linewidth}{
                    \begin{center}
                        \textbf{Adversary} \Adv{}
                    \end{center}

                    \begin{itemize}
                        \item When receiving \PUFE{} from \C{}, simulate \ExtPUF{} between the malicious \Sadv{} and the honest \R{},
                        who uses \PUFE{};
                        \item When \Sadv{} decommits with $s, x, st_\E, p_\E$, output $\msf{Enc}(st), st_\E, p_\E$, where $st$
                        is the one obtained by \R{} in the decommitment.
                    \end{itemize}
                }
            }%
        \end{center}
        \caption{Adversary that breaks the CRP guessing property.}
        \label{fig:adversary extpuf guesscrp}
    \end{figure}

    Therefore, we have shown that if \Sadv{} decommits successfully with $S, X, ST_\E, P_\E$, then with overwhelming probability
    it queries \PUFE{} on some $Q_X$ such that
    \[
        C = ST \oplus \prn{X^n \land R} = \Dec(Q_X) \oplus \prn{X^n \land R}.
    \]

    Likewise, if \Sadv{} can also decommit successfully with $S', Y, ST'_\E, P'_\E$, then with overwhelming probability
    it queries \PUFE{} on some $Q_Y$ such that
    \[
        C = ST' \oplus \prn{Y^n \land R} = \Dec(Q_Y) \oplus \prn{Y^n \land R}.
    \]
    where $ST'$ is the one obtained by \R{} in the decommitment.

    Suppose $X$ and $Y$ differ on an index $J$ (which is also a random variable).
    Then,
    \begin{align*}
        C_{\bm{I}_J} &= \Dec(Q_X)_{\bm{I}_J} \oplus \prn{ X_J^n \land R_{\bm{I}_J} } \\
        &= \Dec(Q_Y)_{\bm{I}_J} \oplus \prn{ Y_J^n \land R_{\bm{I}_J} },
    \end{align*}
    and so
    \[
        R_{\bm{I}_J} = \Dec(Q_X)_{\bm{I}_J} \oplus \Dec(Q_Y)_{\bm{I}_J}.
    \]

    Since the queries $Q_X$ and $Q_Y$ were done by \Sadv{} before receiving $R$, we know $\prn{Q_X, Q_Y} \indep R$.
    Just like in \CPUF{}, this does not necessarily imply $\prn{Q_X, Q_Y} \indep R_{\bm{I}}$,
    but we still have $\prn{Q_X, Q_Y} \indep R_{\bm{I}_j}$ for each $j \in [k]$.
    Therefore, by Lemma \ref{lemma independent}, for each $j \in [k]$,
    \[
        \condavgminentropy{R_{\bm{I}_j}}{Q_X, Q_Y} = \minentropy{R_{\bm{I}_j}} = n,
    \]
    which implies
    \begin{align*}
        &\condavgminentropy{R_{\bm{I}_J}}{\Dec(Q_X)_{\bm{I}_J} \oplus \Dec(Q_Y)_{\bm{I}_J}} \\
        \geq& \ \condavgminentropy{R_{\bm{I}_J}}{Q_X, Q_Y, J} \tag{Lemma \ref{lemma function}} \\
        \geq& \ \condavgminentropy{R_{\bm{I}_J}}{Q_X, Q_Y} - \maxentropy{J} \tag{Lemma \ref{lemma chain rule}} \\
        \geq& \ n - \maxentropy{J} - \maxentropy{J} \tag{Lemma \ref{lemma minentropy}}\\
        =& \ n - 2\log(k).
    \end{align*}
    
    Thus, by Lemma \ref{lemma equality},
    \[
        \prob{R_{\bm{I}_J} = \Dec(Q_X)_{\bm{I}_J} \oplus \Dec(Q_Y)_{\bm{I}_J}} \leq 2^{-n + 2\log(k)},
    \]
    which is negligible.

    \textbf{Extractability:}
    
    Consider the extractor \ext{} depicted in Fig. \ref{fig:extractor extpuf}, which is essentially the same as the original one,
    but adapted to this modified protocol.
    Notice that $x$ is extracted from a query $q$ if and only if
    \begin{itemize}            
        \item $r_{\bm{I}_j} \neq \bm{0}$ for all $j \in [k]$;            
        \item $\msf{check}_x^q$ holds, that is, $c = \mathsf{Dec}(q) \oplus \prn{ x^n \land r }$.
    \end{itemize}

    Furthermore, \ext{} outputs a string $x^*$ if and only if it is the only string $x$ that satisfies
    \[
    \msf{check}_x \coloneq \exists q \in \mathcal{Q} : \msf{check}_x^q.
    \]
    \begin{figure}
        \begin{center}
            \noindent\fbox{%
                \parbox{0.95\linewidth}{
                    \begin{center}
                        \textbf{Extractor \ext}
                    \end{center}
            
                    \ext{} proceeds like an honest \R{}, while also saving \Sadv{}'s queries to \PUFE{} in $\mathcal{Q}$.
                    At the end of the commitment phase, for each $q \in \mathcal{Q}$, it tries to extract a string $x$ from $q$ 
                    by running \texttt{ExtractFromQuery}$(\msf{Dec}(q))$, where:                
                    \pcb[head = \texttt{ExtractFromQuery}$(st)$] {
                        x \coloneqq \varepsilon \\
                        \pcfor j \in [k]: \\
                        \pcind \pcif c_{\bm{I}_j} = st_{\bm{I}_j} \land c_{\bm{I}_j} \neq st_{\bm{I}_j} \oplus r_{\bm{I}_j}: \\
                        \pcind \pcind x \coloneq x \mathbin \Vert 0 \\
                        \pcind \pcelseif c_{\bm{I}_j} = st_{\bm{I}_j} \oplus r_{\bm{I}_j} \land c_{\bm{I}_j} \neq st_{\bm{I}_j}: \\
                        \pcind \pcind x \coloneq x \mathbin \Vert 1 \\
                        \pcind \pcelse : \\
                        \pcind \pcind \pcreturn \bot \\
                        \pcreturn x
                    }

                    Then, it does the following:
                    \begin{itemize}
                        \item If exactly one string $x$ was extracted, output $x^* = x$;
                        \item Otherwise, output $x^* = \bot$.
                    \end{itemize}
                }
            }%
        \end{center}
        \caption{The extractor for \ExtPUF{}.}
        \label{fig:extractor extpuf}
    \end{figure}         
    
    \ext{} is clearly PPT and verifies the simulation property.
    Now, we want to prove the extraction property, that is,
    \[
            \prob{\Sadv{} \text{ decommits successfully to } X \neq X^*} = \varepsilon(n).
    \]

    First, let us consider the probability $\prob{\Sadv{} \text{ decommits successfully to } X, X^* = \bot}$.
    Let \texttt{EXT} be the set of strings \ext{} extracted.
    This event happens in the following cases:
    \begin{itemize}
        \item \textbf{No string was extracted:}
        
        This can be expressed as $\abs{ \texttt{EXT} } = 0$.

        Notice that if $R_{\bm{I}_j} = \bm{0}$ for some $j \in [k]$, then \texttt{ExtractFromQuery} always aborts.
        Given that this happens with negligible probability, we can assume it is not the case.
        
        Suppose \Sadv{} decommits successfully with non-negligible probability with $S, X, ST_\E, P_\E$.
        Let $ST$ be the one obtained by \R{} in the decommitment and $Q = \Enc(ST)$.
        Then, we know $\msf{check}_X^Q$ is verified, because the decommitment is successful.

        Notice that \Sadv{} did not query \PUFE{} with $Q'$ that is close to $Q$,
        because that would imply that $\msf{Dec}(Q') = ST$ and so $\msf{check}_x^{Q'}$ would be true.
        Given that $R_{\bm{I}_j}$ is never $\bm{0}$, this would mean that $X$ was extracted from $Q'$,
        contradicting the assumption.

        Just like we discussed in binding, we will be able to construct an adversary \Adv{} that contradicts the CRP
        guessing property depicted in Fig. \ref{fig:guesscrp}.
        Thus,
        \begin{align*}
            &\prob{\Sadv{} \text{ decommits successfully to } X, \abs{ \texttt{EXT} } = 0} \\
            =& \varepsilon(n).
        \end{align*}

        \item \textbf{Two different strings were extracted:}
        
        This can be expressed as $\abs{ \texttt{EXT} } > 1$.

        Let $J$ be such that $X_J \neq Y_J$ and ${\bm{I}} \coloneqq \bm{I}_J$.
        Then, there exist $Q_X, Q_Y \in \mathcal{Q}$ such that
        \begin{align*}
            C_{\bm{I}} &= \Dec(Q_X)_{\bm{I}} \oplus \prn{ X_J^n \land R_{\bm{I}} } \\
            &= \Dec(Q_Y)_{\bm{I}} \oplus \prn{ Y_J^n \land R_{\bm{I}} },
        \end{align*}
        and so
        \[
            \Dec(Q_X)_{\bm{I}} \oplus \Dec(Q_Y)_{\bm{I}} = R_{\bm{I}}.
        \]
        
        Just like we discussed in binding, this only happens with negligible probability.
        Thus,
        \begin{align*}
            &\prob{\Sadv{} \text{ decommits successfully to } X, \abs{ \texttt{EXT} } > 1} \\
            \leq& \prob{\abs{ \texttt{EXT} } > 1} \\
            =& \varepsilon(n).
        \end{align*}
    \end{itemize}
        
    Finally, let us consider the probability $\prob{\Sadv{} \text{ decommits successfully to } X \neq X^*, X^* \neq \bot}$.
    If this event happens, we know $X^*$ is the only string $x$ for which $\msf{check}_x$ holds.
    
    Now, suppose \Sadv{} decommits successfully with non-negligible probability with $S, X, ST_\E, P_\E$.
    Let $ST$ be the one obtained by \R{} in the decommitment and $Q = \Enc(ST)$.
    Then, we know $\msf{check}_X^Q$ is verified, because the decommitment is successful.

    Just like in the case where no string was extracted, all of this implies that \Sadv{} did not query \PUFE{}
    on some $Q'$ that is close to $Q$, otherwise $X$ would have been extracted from $Q'$ and
    we would be able to construct an adversary \Adv{} that contradicts the CRP
    guessing property depicted in Fig. \ref{fig:guesscrp}.
    Therefore,
    \[
        \prob{\Sadv{} \text{ decommits successfully to } X \neq X^*, X^* \neq \bot} = \varepsilon(n).
    \]

    Thus, we can finally conclude that
    \[
        \prob{\Sadv{} \text{ decommits successfully to } X \neq X^*} = \varepsilon(n).
    \]
\end{proof}
\subsection{Security proof for \CollExtPUF{}, Theorem \ref{collective extpuf one main text}}
\label{appendix:ourproof}

\begin{theorem*}
    \label{collective extpuf one}
    \CollExtPUF{} is an ideal extractable collective commitment in the \FComMPUF{}-hybrid model, with $\kstate = \kout = 0$ and unbounded \kin{}.
\end{theorem*}
\begin{proof}
    \textbf{Computationally Hiding:}
    
    Suppose, by contradiction, that this is not the case.
    Then, there exists an interaction \INTER{} and a malicious receiver \Radv{} such that,
    for the interaction depicted in Fig. \ref{fig:collective hiding},
    \[
        \prob{ \Radv \prn{ \INTER^\CollExtPUF(W) } = W } - \frac{1}{|\Omega_n|}
    \]
    is not negligible.

    Consider a modified protocol \CollExtUnif{},
    \footnote{
        Again, the decommitments of \CollExtUnif{} are actually not defined for $i \in \CLOSED_n$.
        This is not a problem, because we will only run its commitment phase and decommitments for $i \in \OPEN_n$.
    }
    where \S{} uses $u^i \sample \bin^{kn}$ instead of $st^i$, for each $i \in \CLOSED_n$.
    Then,
    \[
        \prob{ \Radv \prn{ \INTER^\CollExtUnif(W) } = W } = \frac{1}{|\Omega_n|},
    \]
    because the distribution of the interaction is independent from $W$.
    Indeed, since \S{} does not abort when \PUFE{} aborts and $\kstate = \kout = 0$, \Radv{} does not
    get any information that depends on $W$.
    
    Now, consider the interaction depicted in Fig. \ref{fig:collective one indist} between a challenger \C{} and a distinguisher \D{},
    corresponding to a generalized indistinguishability property of PUFs.        
    In Fig. \ref{fig:distinguisher collective one cpuf}, we present a distinguisher \D{} that breaks this property,
    thereby contradicting Lemma \ref{collective one indistinguishability}, with $\IND_n = \CLOSED_n$.        
    Here, \CollExtBS{} denotes a protocol just like \CollExtPUF{}, but where:
    \begin{itemize}            
        \item \S{} receives \PUF{} and $\bs^i, p^i$ for each $i \in \CLOSED_n$ from \D{} (who received them from \C{}) and uses each
        $\bs^i$ instead of $st^i$;

        \item \S{} generates the $st^i, p^i$ for each $i \in \OPEN_n$ using \PUF{}.
    \end{itemize}

    Notice that if $b = 0$, then $\bs^i = st^i$ for each $i \in \CLOSED_n$, making \CollExtBS{} identical to \CollExtPUF{}.
    Likewise, if $b = 1$, then $\bs^i = u^i$ for each $i \in \CLOSED_n$, and so \CollExtBS{} is identical to \CollExtUnif{}.
    \begin{figure}[ht]
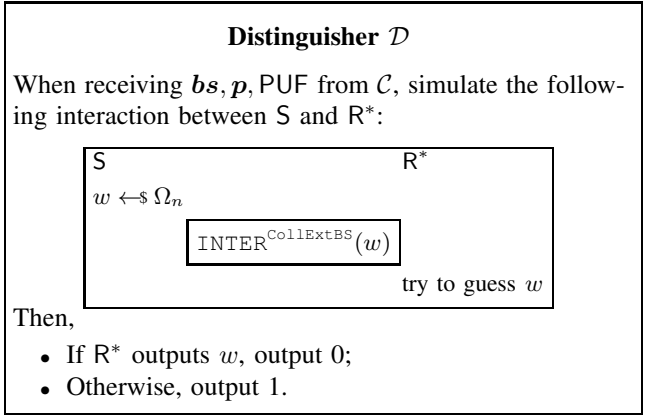

        \centering
        \noindent\fbox{%
            \parbox{0.95\linewidth}{
                \begin{center}
                    \textbf{Distinguisher} $\D$
                \end{center}
                
                When receiving $\bm{\bs}, \bm{p}, \PUF$ from \C{}, simulate the following interaction between \S{} and \Radv{}:
                \procb{}{
                    \S \< \< \Radv \\
                    w \sample \Omega_n \< \< \\
                    \< \pcbox{\INTER^\CollExtBS{}(w)} \< \\
                    \< \< \text{try to guess } w
                }
                Then,
                \begin{itemize}
                    \item If \Radv{} outputs $w$, output 0;
                    \item Otherwise, output 1.
                \end{itemize}
            }
        }%
        \caption{Distinguisher that breaks the generalized indistinguishability property depicted in
        Fig. \ref{fig:collective one indist}.}
        \label{fig:distinguisher collective one cpuf}
    \end{figure}
    
    Thus,
    \begin{equation*}
        \begin{split}
            & 2\prob{\D = B} \\
            =& \condprob{\D = 0}{B = 0} + \condprob{\D = 1}{B = 1} \\
            =& \prob{\Radv \prn{ \INTER^\CollExtPUF{}(W) } = W} + \\
            & \prob{\Radv \prn{ \INTER^\CollExtUnif{}(W) } \neq W} \\
            =& \prob{\Radv \prn{ \INTER^\CollExtPUF{}(W) } = W} + 1 - \frac{1}{|\Omega_n|}, \\
        \end{split}
    \end{equation*}
    and so
    \begin{equation*}
        \begin{split}
            &\prob{\D = B} - \frac{1}{2} \\
            =& \frac{1}{2} \prn{ \prob{\Radv \prn{ \INTER^\CollExtPUF{}(W) } = W} - \frac{1}{|\Omega_n|} } \\
        \end{split}
    \end{equation*}
    is non-negligible, which is a contradiction.
    
    \textbf{Statistically Binding:}
    
    Suppose a malicious sender \Sadv{} and an honest receiver interact according to some \INTER{}, which involves \Sadv{} making a commitment.
    We have to show that the probability of \Sadv{} being able to open some commitment $I$ successfully to strings $X^I$ and $Y^I$
    that differ on an index $J$ is negligible.

    We want to show that
    \[
        \prob{R^I_{\bm{I}_J} = \Dec(Q_X)_{\bm{I}_J} \oplus \Dec(Q_Y)_{\bm{I}_J}}
    \]
    is negligible.

    Notice that
    \begin{align*}
        &\condavgminentropy{R^I_{\bm{I}_J}}{\Dec(Q_X)_{\bm{I}_J} \oplus \Dec(Q_Y)_{\bm{I}_J}} \\
        \geq& \ \condavgminentropy{R^I_{\bm{I}_J}}{Q_X, Q_Y, J} \tag{Lemma \ref{lemma function}} \\
        \geq& \ \condavgminentropy{R^I_{\bm{I}_J}}{Q_X, Q_Y} - \maxentropy{J} \tag{Lemma \ref{lemma chain rule}} \\
        \geq& \ n - \maxentropy{I} - \maxentropy{J} - \maxentropy{J} \tag{Lemma \ref{lemma minentropy}}\\
        =& \ n - \log(N(n)) - 2\log(k(n)).
    \end{align*}
    
    Thus, by Lemma \ref{lemma equality},
    \begin{align*}
        &\prob{R^I_{\bm{I}_J} = \Dec(Q_X)_{\bm{I}_J} \oplus \Dec(Q_Y)_{\bm{I}_J}} \\
        \leq& 2^{ -n + \log(N(n)) + 2\log(k(n)) },
    \end{align*}
    which is negligible, since $N(n)$ and $k(n)$ are polynomial.
    
    \textbf{Extractability:}

    Here, we use the extractor from Theorem \ref{single nocomm extpuf} for each string being committed,
    as depicted in Fig. \ref{fig:extractor collective extpuf}.
    \begin{figure}[ht]
        \noindent\fbox{%
            \parbox{0.95\linewidth}{
                \begin{center}
                    \textbf{Extractor \ext}
                \end{center}
        
                \ext{} proceeds like an honest \R{}, while also saving \Sadv{}'s queries to \PUFE{} in $\mathcal{Q}$.
                At the end of the commitment phase, for each $i \in [N]$ and $q \in \mathcal{Q}$,
                it tries to extract a string $x^i$ from $c^i$ and $q$ by running $\texttt{ExtractFromQuery}^i(\msf{Dec}(q))$, where:                
                \pcb[head = $\texttt{ExtractFromQuery}^i(st)$] {
                    x^i \coloneqq \varepsilon \\
                    \pcfor j \in [k]: \\
                    \pcind \pcif c^i_{\bm{I}_j} = st_{\bm{I}_j} \land c^i_{\bm{I}_j} \neq st_{\bm{I}_j} \oplus r^i_{\bm{I}_j}: \\
                    \pcind \pcind x^i \coloneq x^i \mathbin \Vert 0 \\
                    \pcind \pcelseif c^i_{\bm{I}_j} = st_{\bm{I}_j} \oplus r^i_{\bm{I}_j} \land c^i_{\bm{I}_j} \neq st_{\bm{I}_j}: \\
                    \pcind \pcind x^i \coloneq x^i \mathbin \Vert 1 \\
                    \pcind \pcelse : \\
                    \pcind \pcind \pcreturn \bot \\
                    \pcreturn x^i
                }

                Then, output $\bm{x}^*$, where for each $i \in [N]$:
                \begin{itemize}
                    \item If exactly one string $x^i$ was extracted from $c^i$, ${x^*}^i \coloneqq x^i$;
                    \item Otherwise, ${x^*}^i \coloneqq \bot$.
                \end{itemize}
            }
        }%
        \caption{The extractor for \CollExtPUF{}.}
        \label{fig:extractor collective extpuf}
    \end{figure}
    We want to show that
    \[
        \prob{ \Sadv \text{ decommits some $I$ successfully to } X^i \neq \prn{X^*}^i }
    \]
    is negligible.

    This is completely analogous to what was done in Theorem \ref{single nocomm extpuf}.
    Indeed, as the extractability proof of \ExtPUF{} relied on its binding proof, the same happens
    for \CollExtPUF{}.
\end{proof}
\subsection{Proof of UC security}
\label{appendix:ucproof}

The bit commitment functionality is depicted in Fig. \ref{fig:bit commit functionality}.
\begin{figure}[ht]
    \begin{center}
        \noindent\fbox{%
            \parbox{0.9\linewidth}{
                \begin{center}
                    \textbf{Bit commitment functionality \Fcom}
                \end{center}
        
                \textbf{Commitment phase:}
                
                Upon receiving a message $(\msf{commit}, b)$ from \Sideal{}, where $b \in \bin$,
                record the value $b$ and send the message \msf{commit} to \Rideal{} and \Sim{}. 
                Ignore any subsequent \msf{commit} messages.\\
        
                \textbf{Decommitment phase:}
                
                Upon receiving a value \msf{open} from \Sideal{}, if some value $b$ was previously recorded,
                then send the message $(\msf{open}, b)$ to \Rideal{} and \Sim{} and halt.
            }
        }%
    \end{center}
    \caption{Bit commitment functionality \Fcom{}.}
    \label{fig:bit commit functionality}
\end{figure}

Consider the protocol given by \texttt{UCCompiler} in Fig. \ref{fig:revised uccompiler} with an ideal 
extractable collective commitment \collcom{}.
We prove that this protocol UC-realizes \Fcom{}.
\begin{lemma}
    \label{ucproofreceiver}
    Let \collcom{} be an ideal extractable collective commitment scheme in the \FComMPUF-hybrid model.
    Then, for any real world adversary \Adv{} that corrupts the receiver, there exists a simulator \Sim{} such that no environment
    \Z{} can distinguish between the corresponding real world and ideal world processes.
\end{lemma}
\begin{proof}
    Consider the simulator defined in Fig. \ref{fig:simulator hd}.
    We aim to show that no environment \Z{} can distinguish the real world from the ideal world by using a hybrid argument,
    just like in \cite{UCComm}. 
    Starting from the real world, we define a sequence of hybrids that gradually transition to the ideal world.  
    Each intermediate hybrid is defined within a modified version of the ideal world, where the simulator has access to the input $b_\Z$
    chosen by \Z{} for \S{}.
    Furthermore, we construct each hybrid's simulator based on the previous one and prove that \Z{} cannot distinguish between consecutive hybrids.  
    Since the final hybrid corresponds to the ideal world, using the simulator defined in Fig. \ref{fig:simulator hd},
    the result follows.
    \begin{figure}[ht]
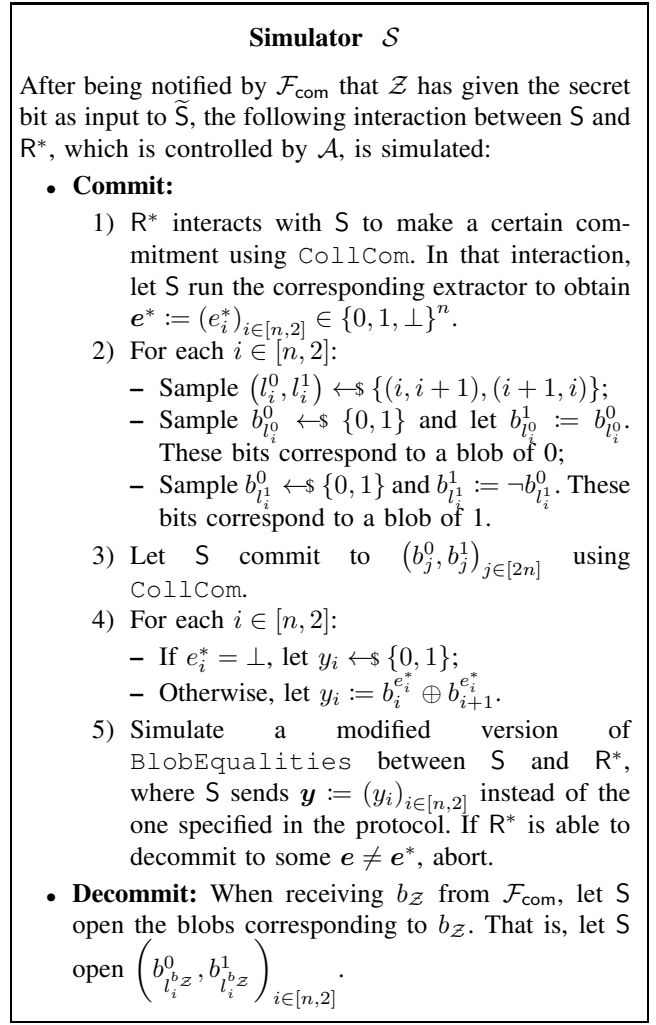

        \centering
        \noindent\fbox{%
            \parbox{0.95\linewidth}{
                \begin{center}
                    \textbf{Simulator } \Sim{}
                \end{center}

                After being notified by \Fcom{} that \Z{} has given the secret bit as input to \Sideal{},
                the following interaction between \S{} and \Radv{}, which is controlled by \Adv{}, is simulated:
                \begin{itemize}
                    \item \textbf{Commit:}
                        \begin{enumerate}
                            \item \Radv{} interacts with \S{} to make a certain commitment using \collcom{}.
                            In that interaction, let \S{} run the corresponding extractor to obtain
                            $\bm{e}^* \coloneqq \prn{e_i^*}_{i \in [n, 2]} \in \set{ 0, 1, \bot}^n$.
                            
                            \item For each $i \in [n, 2]$:
                            \begin{itemize}
                                \item Sample $\prn{l_i^0, l_i^1} \sample \set{(i, i+1), (i+1, i)}$;
                                
                                \item Sample $b_{l_i^0}^0 \sample \bin$ and let $b_{l_i^0}^1 \coloneqq b_{l_i^0}^0$.
                                These bits correspond to a blob of 0;
                                
                                \item Sample $b_{l_i^1}^0 \sample \bin$ and $b_{l_i^1}^1 \coloneqq \neg b_{l_i^1}^0$.
                                These bits correspond to a blob of 1.
                            \end{itemize}

                            \item Let \S{} commit to $\prn{b_j^0, b_j^1}_{j \in [2n]}$ using \collcom{}.
                            
                            \item For each $i \in [n, 2]$:
                            \begin{itemize}
                                \item If $e_i^* = \bot$, let $y_i \sample \bin$;
                                \item Otherwise, let $y_i \coloneqq b_i^{e_i^*} \oplus b_{i + 1}^{e_i^*}$.
                            \end{itemize}

                            \item Simulate a modified version of \texttt{BlobEqualities} between \S{} and \Radv{},
                            where \S{} sends $\bm{y} \coloneqq \prn{y_i}_{i \in [n, 2]}$ instead of the one specified in the protocol.
                            If \Radv{} is able to decommit to some $\bm{e} \neq \bm{e}^*$, abort.
                        \end{enumerate}
                    
                    \item \textbf{Decommit:}
                    When receiving $b_\Z$ from \Fcom{}, let \S{} open the blobs corresponding to $b_\Z$.
                    That is, let \S{} open $\prn{ b_{l_i^{b_\Z}}^0, b_{l_i^{b_\Z}}^1 }_{i \in [n, 2]}$.
                \end{itemize}
            }
        }%
        \caption{Simulator for the case where the receiver is dishonest.}
        \label{fig:simulator hd}
    \end{figure}

    Consider the following hybrids:
    \begin{itemize}
        \item \textbf{Hybrid $H_0$:}
        This is the real world execution of the protocol \texttt{UCCompiler}.
        
        \item \textbf{Hybrid $H_1$:}
        This hybrid is in the modified ideal world defined above, where the simulator $\Sim_1$ simulates an execution of the real world process,
        using $b_\Z$ as input for \S{}.
        Since $H_1$ is just the real world process executed through the simulator $\Sim_1$, hybrids $H_0$ and $H_1$ are identical.

        \item \textbf{Hybrid $H_2$:}
        In this hybrid, consider the simulator $\Sim_2$ that runs steps 1 and 5 of \Sim{}.
        Notice that, from the extractability property of \collcom{}, the extractor simulates an honest receiver (which, in this case, is \S{})
        and extracts $\bm{e}^*$ such that \Radv{} is only able to decommit to some $\bm{e} \neq \bm{e}^*$ with negligible probability.
        Furthermore, since for each $i \in [n, 2]$, the blobs $\mathbf{B}_i$ and $\mathbf{B}_{i+1}$ have the same value, we have
        $b_i^{e_i} \oplus b_{i+1}^{e_i} = b_i^0 \oplus b_{i+1}^0$.
        Thus, this is indistinguishable to $H_1$.

        \item \textbf{Hybrid $H_3$:}
        In this hybrid, consider the simulator $\Sim_3$ that generates its $l_i^0$ and $l_{i+1}^1$ for each $i \in [n, 2]$ during the commitment phase.
        Furthermore, in the decommitment phase, $\Sim_3$ lets \S{} open the blobs of indices $l_i^{b\Z}$.
        This is identical to $H_2$.

        \item \textbf{Hybrid $H_4$:}
        In this hybrid, consider the simulator \Sim{} that we defined earlier.
        Notice that we are now in the regular ideal world, where \Sim{} no longer has access to $b_\Z$.
        This hybrid is not identical to $H_3$, since the blobs no longer have the same value.
        However, we will show that no environment \Z{} can distinguish between them.

        Consider the interaction depicted in Fig. \ref{fig:indist h3 h4} between a challenger \C{} and \Z{}, in which \Z{} attempts to distinguish
        between $H_3$ and $H_4$.
        \begin{figure}[ht]
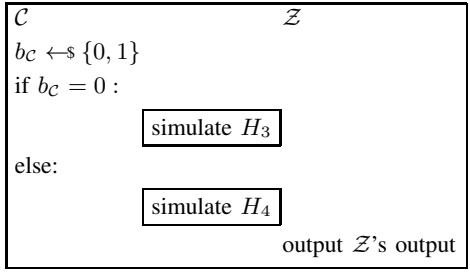

            \procb{}{
                \C \< \< \Z \\
                b_\C \sample \bin \< \< \\
                \text{if } b_\C = 0: \< \< \\
                \< \pcbox{ \text{simulate } H_3 } \< \\
                \text{else:} \< \< \\
                \< \pcbox{ \text{simulate } H_4 } \< \\
                \< \< \text{output \Z{}'s output}
            }
            \caption{Interaction corresponding to indistinguishability between hybrids $H_3$ and $H_4$.}
            \label{fig:indist h3 h4}
        \end{figure}

        Assume, by contradiction, that there exists \Z{} such that $\prob{\Z = B_\C} - \frac{1}{2}$ is non-negligible.
        Since the only difference between the two hybrids lies in the values of the bits being committed by \S{},
        we will be able to construct an interaction that breaks the computational hiding property of \collcom{}, thus 
        showing that both hybrids are indeed indistinguishable.
        
        Consider the interaction depicted in Fig. \ref{fig:reduction h3 h4} between \S{} and \Radv{}, where \INTER{} is
        defined in Fig. \ref{fig:inter h3 h4}.
        \begin{figure}[ht]
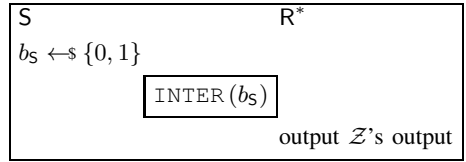

            \procb{}{
                \S \< \< \Radv \\
                b_\S \sample \bin \< \< \\
                \< \pcbox{ \INTER \prn{b_\S} } \< \\
                \< \< \text{output \Z{}'s output}
            }
            \caption{Reduction to the computational hiding property of \collcom{}.}
            \label{fig:reduction h3 h4}
        \end{figure}
        \begin{figure}[ht]
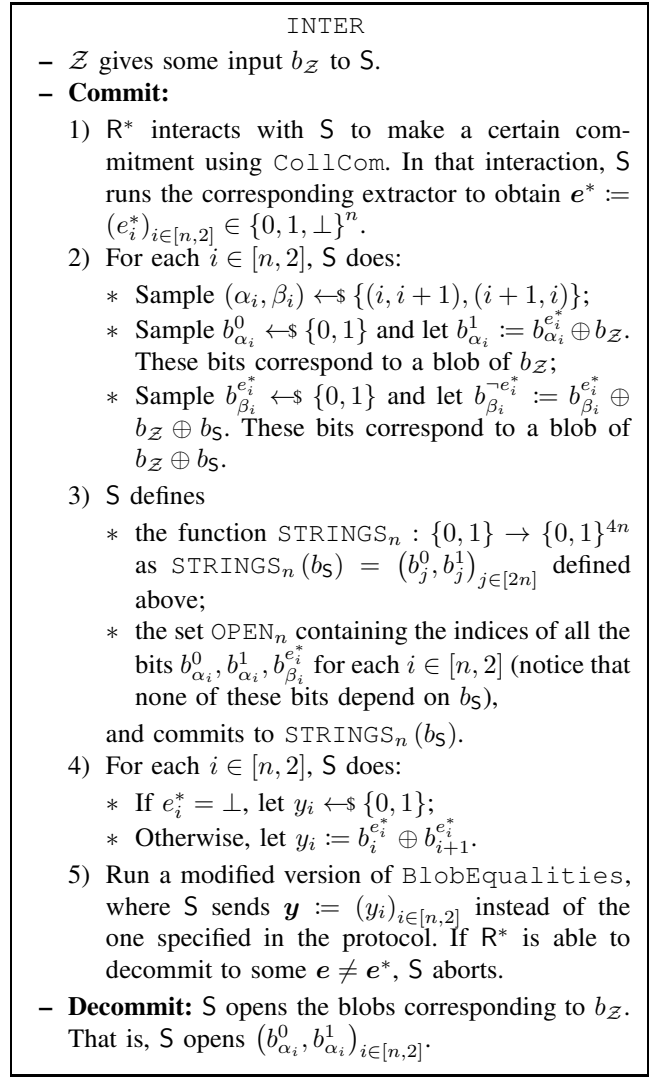

            \centering
            \noindent\fbox{%
                \parbox{0.95\linewidth}{
                    \begin{center}
                        \INTER
                    \end{center} 
                    \begin{itemize}
                        \item \Z{} gives some input $b_\Z$ to \S{}.
                        \item \textbf{Commit:}
                            \begin{enumerate}
                                \item \Radv{} interacts with \S{} to make a certain commitment using \collcom{}.
                                In that interaction, \S{} runs the corresponding extractor to obtain
                                $\bm{e}^* \coloneqq \prn{e_i^*}_{i \in [n, 2]} \in \set{ 0, 1, \bot}^n$.
                                
                                \item For each $i \in [n, 2]$, \S{} does:
                                \begin{itemize}
                                    \item Sample $\prn{\alpha_i, \beta_i} \sample \set{(i, i+1), (i+1, i)}$;
                                    \item Sample $b_{\alpha_i}^0 \sample \bin$ and let $b_{\alpha_i}^1 \coloneqq b_{\alpha_i}^{e_i^*} \oplus b_\Z$.
                                    These bits correspond to a blob of $b_\Z$;                                        
                                    \item Sample $b_{\beta_i}^{e_i^*} \sample \bin$ and let $b_{\beta_i}^{\neg e_i^*} \coloneqq b_{\beta_i}^{e_i^*} \oplus b_\Z \oplus b_\S$.
                                    These bits correspond to a blob of $b_\Z \oplus b_\S$.
                                \end{itemize}

                                \item \S{} defines
                                \begin{itemize}
                                    \item the function $\STRINGS_n : \bin \to \bin^{4n}$ as $\STRINGS_n \prn{ b_\S } = \prn{b_j^0, b_j^1}_{j \in [2n]}$
                                    defined above;
                                    \item the set $\OPEN_n$ containing the indices of all the bits $b_{\alpha_i}^0, b_{\alpha_i}^1, b_{\beta_i}^{e_i^*}$
                                    for each $i \in [n, 2]$ (notice that none of these bits depend on $b_\S$),
                                \end{itemize}
                                and commits to $\STRINGS_n \prn{ b_\S }$.
                                
                                \item For each $i \in [n, 2]$, \S{} does:
                                \begin{itemize}
                                    \item If $e_i^* = \bot$, let $y_i \sample \bin$;
                                    \item Otherwise, let $y_i \coloneqq b_i^{e_i^*} \oplus b_{i + 1}^{e_i^*}$.
                                \end{itemize}

                                \item Run a modified version of \texttt{BlobEqualities}, where \S{} sends $\bm{y} \coloneqq \prn{y_i}_{i \in [n, 2]}$
                                instead of the one specified in the protocol.
                                If \Radv{} is able to decommit to some $\bm{e} \neq \bm{e}^*$, \S{} aborts.
                            \end{enumerate}
                        
                        \item \textbf{Decommit:}
                        \S{} opens the blobs corresponding to $b_\Z$.
                        That is, \S{} opens $\prn{ b_{\alpha_i}^0, b_{\alpha_i}^1 }_{i \in [n, 2]}$.
                    \end{itemize}
                }
            }%
            \caption{Definition of the interaction \INTER{}.}
            \label{fig:inter h3 h4}
        \end{figure}
        Notice how important it is for \Radv{} to commit before \S{},
        since the bits \S{} commits to depend on what it extracts from \Radv{}'s commitment.        
        Furthermore, notice that for each value of $b_\S$, the only thing that differs is the value of the blobs.
        Indeed, when $b_\S = 0$, all the blobs have value $b_\Z$, causing the simulation to be distributed like in $H_3$.
        On the other hand, when $b_\S = 1$, there are blobs of 0 and 1, and so it is distributed like in $H_4$.
        Therefore, the probability of \Radv{} correctly guessing $b_\S$ is
        \begin{align*}
            & 2\prob{\Radv \prn{ \INTER \prn{B_\S} } = B_\S} \\
            =& \condprob{\Radv \prn{ \INTER (0)} = 0}{B_\S = 0} + \\
            &\condprob{\Radv \prn{ \INTER (1)} = 1}{B_\S = 1} \\
            =& \condprob{\Z = 0}{B_\C = 0} + \condprob{\Z = 1}{B_\C = 1} \\
            =& 2\prob{\Z = B_\C}
        \end{align*}
        and thus
        \[
            \prob{\Radv \prn{ \INTER (B_\S)} = B_\S} - \frac{1}{2} = \prob{\Z = B_\C} - \frac{1}{2},
        \]
        which we assumed to be non-negligible.
        This contradicts the computational hiding property of \collcom{}.
    \end{itemize}
\end{proof}

\begin{lemma}
    \label{ucproofsender}
    Let \collcom{} be an ideal extractable collective commitment scheme in the \FComMPUF-hybrid model.
    Then, for any real world adversary \Adv{} that corrupts the sender, there exists a simulator \Sim{} such that no environment
    \Z{} can distinguish between the corresponding real world and ideal world processes.
\end{lemma}
\begin{proof}
    Consider the simulator defined in Fig. \ref{fig:simulator dh}.
    \begin{figure}[ht]
        \centering
        \noindent\fbox{%
            \parbox{0.95\linewidth}{
                \begin{center}
                    \textbf{Simulator } \Sim{}
                \end{center}

                Simulate the interaction between \Sadv{} and \R{} as follows:
                \begin{itemize}
                    \item \textbf{Commit:}
                        \begin{enumerate}
                            \item For each $i \in [n, 2]$:
                            \begin{itemize}
                                \item Sample $e_i \sample \bin$.
                            \end{itemize}
                            
                            \item Let \R{} commit to $\bm{e} \coloneqq \prn{e_i}_{i \in [n, 2]}$ using \collcom{}.
                        
                            \item \Sadv{} interacts with \R{} to make a certain commitment using \collcom{}.
                            In that interaction, let \R{} run the corresponding extractor to obtain
                            $\prn{{b^*}_j^0, {b^*}_j^1}_{j \in [2n]} \in \set{ 0, 1, \bot}^{2n}$.

                            \item For each $j \in [2n]$:
                            \begin{itemize}
                                \item If ${b^*}_j^0 = \bot$ or ${b^*}_j^1 = \bot$, let ${b^*}_j \coloneqq \bot$;
                                \item Otherwise, let ${b^*}_j \coloneqq {b^*}_j^0 \oplus {b^*}_j^1$.
                            \end{itemize}
                            
                            \item Simulate \texttt{BlobEqualities} between \Sadv{} and \R{}.
                            
                            \item Consider the sets $A_i \coloneqq \set{{b^*}_i, {b^*}_{i+1}}$ for each $i \in [n, 2]$ and let
                            $A \coloneqq \bigcap_{i \in [n, 2]} A_i \setminus \set{ \bot }$.
                            This set represents to which bits \Sadv{} can decommit to in the decommitment phase.
                            Check which of the following cases applies:                            
                            \begin{enumerate}
                                \item If $A = \emptyset$, let \Sideal{} send $b^* \coloneqq 0$ to \Fcom{};                                    
                                \item If $A = \set{b^*}$ for some $b^*$, let \Sideal{} send $b^*$ to \Fcom{};
                                \item If $A = \bin$, abort.
                            \end{enumerate}
                        \end{enumerate}
                    
                    \item \textbf{Decommit:}
                        If \Sadv{} correctly decommits to some $b$, then
                        \begin{itemize}
                            \item If $b = b^*$, let \Sideal{} send \msf{open} to \Fcom{};
                            \item Otherwise, abort.
                        \end{itemize}                
                \end{itemize}
            }
        }%
        \caption{Simulator for the case where the sender is dishonest.}
        \label{fig:simulator dh}
    \end{figure}
    Notice that the real and ideal world processes are almost identical.
    Indeed, \Sim{} lets \R{} run the extractor for \collcom{}, which simulates an honest receiver (which in this case is \R{}),
    and then it lets \R{} honestly follow the protocol \texttt{BlobEqualities}.
    The only difference is that it additionally aborts in the following cases:
    \begin{itemize}
        \item In step 6, if $A = \bin$, meaning $A_i = \bin$ for all $i \in [n, 2]$.
        This means that \Sadv{} managed to cheat the $\texttt{BlobEqualities}$ protocol, which we are going to show
        only happens with negligible probability.
        
        Due to extractability property of \collcom{}, we know that with overwhelming probability, the commitments of each
        $b_j^e$ can only be opened to ${b^*}_j^e$.
        Consider the function $f : \bin^n \to \bin^n$ such that
        $f(\bm{e}) = \prn{f_i(e_i)}_{i \in [n,2]} = \prn{{b^*}_i^{e_i} \oplus {b^*}_{i+1}^{e_i}}_{i \in [n,2]}$,
        where $\bm{e}$ is indexed in $[n,2]$.
        Notice that for each $i \in [n, 2]$     
        \[
            A_i = \set{ {b^*}_i^0 \oplus {b^*}_i^1, {b^*}_{i+1}^0 \oplus {b^*}_{i+1}^1 } = \bin.
        \]
        
        Furthermore,
        \begin{align*}
            &\forall i \in [n, 2] \ \set{ {b^*}_i^0 \oplus {b^*}_i^1, {b^*}_{i+1}^0 \oplus {b^*}_{i+1}^1 } = \bin \\
            \iff& \forall i \in [n, 2] \ {b^*}_i^0 \oplus {b^*}_i^1 \neq {b^*}_{i+1}^0 \oplus {b^*}_{i+1}^1 \\
            \iff& \forall i \in [n, 2] \ {b^*}_i^0 \oplus {b^*}_i^0 \neq {b^*}_{i+1}^1 \oplus {b^*}_{i+1}^1 \\
            \iff& \forall i \in [n, 2] \ f_i(0) \neq f_i(1) \\
            \iff& \forall i \in [n, 2] \ f_i \text{ is a bijection in } \bin \\
            \implies& \forall i \in [n, 2] \ f_i \circ f_i = \msf{id}_{\bin} \\
            \implies& f \circ f = \msf{id}_{\bin^n}.
        \end{align*}

        Now, observe that \Sadv{} successfully passed \texttt{BlobEqualities} and the commitments of each $b_j^e$ were opened to ${b^*}_j^e$.
        This means that for each $i \in [n, 2]$ it sent $y_i = {b^*}_i^{e_i} \oplus {b^*}_{i+1}^{e_i} = f_i(e_i)$.
        In other words, it sent $\bm{y} = f(\bm{e})$, and since $f \circ f = \msf{id}_{\bin^n}$, we know
        $\bm{y} = f(\bm{e}) \iff f(\bm{y}) = \bm{e}$.
        
        This implies that \Sadv{} could successfully guess the $\bm{e}$ generated by \R{}.
        However, this only happens with negligible probability due to the computational hiding property of \collcom{}.

        More specifically, consider the interaction depicted in Fig. \ref{fig:reduction dh} between \R{} and \Sadv{},
        where \INTER{} denotes an execution of \texttt{UCCompiler} in which \R{} commits to $\bm{e}$ using \collcom{}
        and $\bm{y}$ denotes the one \Sadv{} sends in \texttt{BlobEqualities}.
        \begin{figure}[ht]
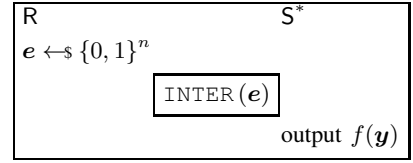

            \procb{}{
                \R \< \< \Sadv \\
                \bm{e} \sample \bin^n \< \< \\
                \< \pcbox{ \INTER \prn{\bm{e}} } \< \\
                \< \< \text{output } f(\bm{y})
            }
            \caption{Reduction to the computational hiding property of \collcom{}.}
            \label{fig:reduction dh}
        \end{figure}
        Then, the probability that the simulation aborted in this step is
        \begin{align*}
            \prob{\bm{Y} = f(\bm{E})} &= \prob{f(\bm{Y}) = \bm{E}} \\
            &= \prob{\Sadv \prn{\INTER(\bm{E})} = \bm{E}},
        \end{align*}
        which is negligible by the computational hiding property of \collcom{}.

        \item In the decommitment phase, when \Sadv{} opens $b \neq b^*$.
        Due to the extractability property of \collcom{}, this only happens with negligible probability.
    \end{itemize}

    Thus, we conclude that the simulator only additionally aborts with negligible probability.
    This, along with the fact that the interaction simulated by \Sim{} is otherwise identical to the real-world process,
    implies that no environment \Z{} can distinguish between the real and ideal worlds.
\end{proof}

Thus, from Lemmas \ref{ucproofreceiver} and \ref{ucproofsender}, we can finally conclude the following:
\begin{theorem}
\label{thm:ucappendix}
    Let \collcom{} be an ideal extractable collective commitment scheme in the \FComMPUF-hybrid model.
    Then, the corresponding protocol given by \texttt{UCCompiler} UC-realizes \Fcom{}
    in the \FComMPUF-hybrid model.
\end{theorem}

\end{document}